\documentclass[journal,12pt,onecolumn,draftclsnofoot]{IEEEtran}

\usepackage{cite}
\usepackage{hyperref}
\usepackage{amsmath,amssymb,amsfonts}
\usepackage{graphicx}
\usepackage{multicol}
\usepackage{textcomp}
\usepackage{caption}
\usepackage{subcaption}
\usepackage{optidef}
\usepackage{bbm}
\usepackage{bm}
\usepackage{empheq}
\usepackage{amsthm}
\usepackage{algorithm}
\usepackage[noend]{algpseudocode}
\usepackage{mathtools,amssymb}

\usepackage{xcolor}
\def\BibTeX{{\rm B\kern-.05em{\sc i\kern-.025em b}\kern-.08em
    T\kern-.1667em\lower.7ex\hbox{E}\kern-.125emX}}
    
\newtheorem{theorem}{Theorem}
\newtheorem{lemma}{Lemma}

\newtheorem{remark}{Remark}
\newtheorem{proposition}{Proposition}
\newtheorem{corollary}{Corollary}
    
\begin{document}

\title{Minimizing Age of Incorrect Information for Unreliable Channel with Power Constraint\\
}

\author{%
\IEEEauthorblockN{Yutao Chen and Anthony Ephremides}\\
    \IEEEauthorblockA{Department of Electrical and Computer Engineering, University of Maryland}
}

\maketitle

\begin{abstract}
Age of Incorrect Information (AoII) is a newly introduced performance metric that considers communication goals. Therefore, comparing with traditional performance metrics and the recently introduced metric - Age of Information (AoI), AoII achieves better performance in many real-life applications. However, the fundamental nature of AoII has been elusive so far. In this paper, we consider the AoII in a system where a transmitter sends updates about a multi-state Markovian source to a remote receiver through an unreliable channel. The communication goal is to minimize AoII subject to a power constraint. We cast the problem into a Constrained Markov Decision Process (CMDP) and prove that the optimal policy is a mixture of two deterministic threshold policies. Afterward, by leveraging the notion of Relative Value Iteration (RVI) and the structural properties of threshold policy, we propose an efficient algorithm to find the threshold policies as well as the mixing coefficient. Lastly, numerical results are laid out to highlight the performance of AoII-optimal policy.
\end{abstract}

\section{Introduction}
Applications, such as autonomous vehicles, control systems, and unmanned aerial vehicles (UAVs), rely heavily on the exchange of time-sensitive information. In these applications, the freshness of information is critical. Conventional metrics such as throughput and delay are not always optimal when considering the freshness of information. The Age of Information (AoI) introduced in \cite{b1} offers a new way to quantify the freshness of information. Let $W_t$ be the generation time of the last received packet. AoI is a function defined by $\Delta_{AoI}(t)=t-W_t$.

Recently, research on AoI have been growing fast \cite{b2,b3,b4}. As AoI gives priority to the updates that can greatly reduce the information time lag at the destination, it provides performance improvement in many applications, especially when the freshness of information is important \cite{b5}. However, in real-life applications, the communication goals vary, and it is not always the goal to keep the information at the destination as fresh as possible. For example, in temperature monitoring, the communication goal is to monitor the abnormal temperature fluctuation of the system and quickly respond to temperature abnormalities. Thus, the freshness of information is not the only priority. We also need to monitor abnormalities as large abnormalities are harmful to the system even when they are new.

Noticing the shortcomings of AoI, researchers proposed several variations on the notion of age. In \cite{b6}, age of synchronization (AoS) is proposed in the framework of web caching. Value of Information of Update (VoIU) is proposed in \cite{b7} which captures the degree of importance of the information received at the destination. In \cite{b8} and \cite{b9}, the concept of effective age is proposed aiming to make connections between age and estimation error. The authors in \cite{b10} introduce the metric Urgency of Information (UoI) which considers the context of information.

To be even more adaptable to various communication goals, a new metric Age of Incorrect Information (AoII) is introduced in \cite{b11}. It captures well not only the freshness of information but also the information content of the transmitted packets and the knowledge at the destination. Several works have been done since the introduction. The authors in \cite{b12} study the AoII with general time function and provide several real-life applications to highlight the advantages of AoII over AoI and the error-based measure approach. In \cite{b13}, extensive numerical results are laid out to compare the performances of AoII and other performance metrics under different policies. However, the considered communication model and the chosen dissatisfaction functions in these papers are simple. Thus, the performance of more general AoII in a more complicated system is still unclear. In this paper, we study the system where the source is modeled by a multi-state Markov chain and adopt the AoII that considers the quantified mismatch between the source and the knowledge at the destination.

At the same time, we investigate the optimization problem in the presence of power constraints. Similar constraints are considered in \cite{b21}, with the goal of minimizing AoI. Our problem is more complicated because AoI ignores the information content of the transmitted updates. \cite{b22,b23,b24,b25} study the problem of remote estimation under resource constraints. However, they focus mainly on minimizing the estimation error but ignore the effect of time as persistent errors will cause more harm to the system than short-lived errors in many real-life applications.

The main contributions of this work are: 1) We adopt the AoII that considers the quantified mismatch and model the source using a multi-state Markov chain. 2) We study the minimization of AoII under power constraints. 3) We rigorously prove the structural properties of the optimal policy and propose an efficient algorithm to obtain the optimal policy.

The rest of this paper is organized as follows: In section \ref{sec-SystemOverview}, we discuss the communication model and the system dynamic under the chosen AoII. Section \ref{sec-ProblemOptimization} presents a step-by-step analysis of the optimization problem and introduces the proposed algorithm. Lastly, in Section \ref{sec-NumericalResults}, numerical results are laid out.

\section{System Overview}\label{sec-SystemOverview}
\subsection{Communication Model}\label{sec-SystemModel}
We consider a slotted-time system in which a transmitter sends updates about a process to a remote receiver through an unreliable channel. The transmitted update will not be corrupted during the transmission but the transmission will not necessarily succeed. When the transmission fails, the update will be discarded and it will not affect the transmitter's decision at the next time slot. We denote the channel realization as $r_t$ where $r_t = 1$ if the transmission succeeds and $r_t = 0$ otherwise. We assume $r_t$ is independent and identically distributed over the time slots. We define $Pr(r_t = 1) = p_s$ and $Pr(r_t = 0) = 1-p_s = p_f$. We notice that, in many status-update systems, the size of the update is very small so that the transmission time for an update is much smaller than the time unit used to measure the dynamic of the process. Thus, when a transmission attempt succeeds, the update is assumed to be received instantly by the receiver. This assumption will provide us with analytical benefits, and a similar assumption is also made in \cite{b14}. The source process $\{X_t\}_{t\in\mathbbm{N}}$ is modeled by an N-state Markov chain where transmissions only happen between adjacent states with probability $2p$ and themselves with probability $1-2p$. An illustration of the Markovian source is shown in Fig. \ref{fig-MarkovianSource}.

\begin{figure*}%
\centering
\begin{subfigure}{\columnwidth}
\includegraphics[width=\textwidth]{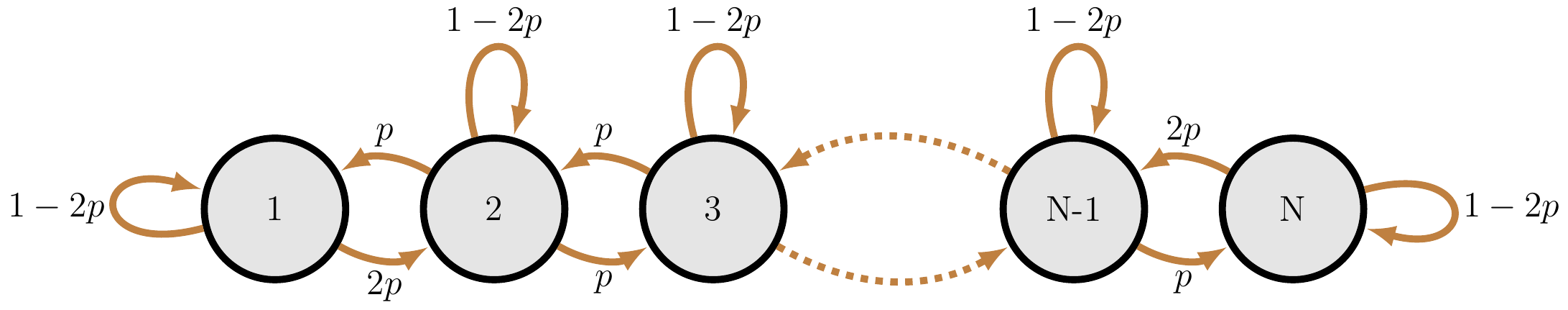}%
\caption{The Markovian source.}%
\label{fig-MarkovianSource}%
\end{subfigure}\hfill%
\begin{subfigure}{\columnwidth}
\includegraphics[width=\textwidth]{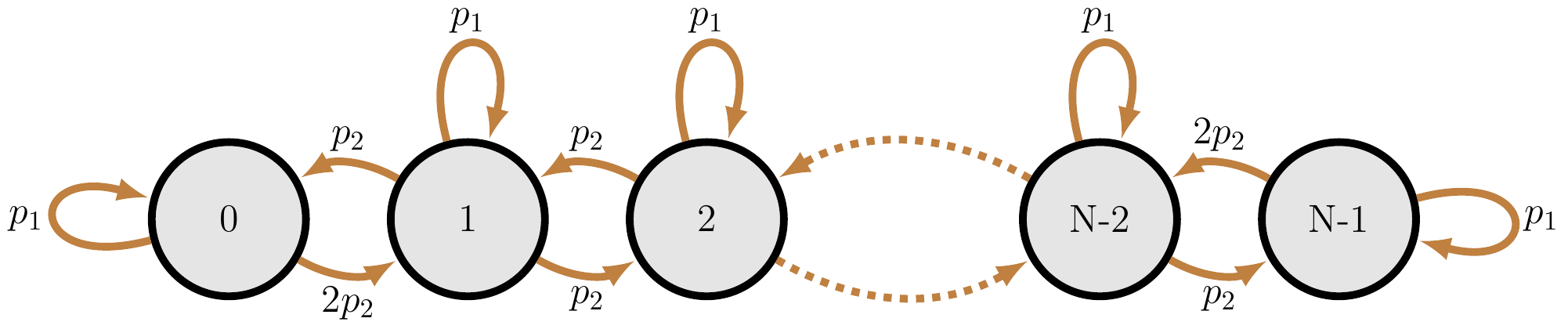}%
\caption{The evolution of $d$.}%
\label{fig-EvolutionD}%
\end{subfigure}\hfill
\caption{Illustrations of the Markovian source and the evolution of $d$.}
\end{figure*}

The transmitter is capable of generating update $X_t$ by sampling the process at any time on its own will. However, the sampling opportunities only occur at the beginning of each time slot. We assume the transmitter is also capable of making transmission attempt in the same time slot the sampling happens. Every time the transmission succeeds, the receiver will use the received update as its new estimate $\hat{X}_t$. The receiver will send an $ACK/NACK$ packet to inform the transmitter whether it has received a new update. We suppose that the $ACK/NACK$ packets will be delivered reliably, and the transmission time is negligible as the packets are very small in general. Therefore, if $ACK$ is received, the transmitter knows that the transmission succeeded, and the receiver's estimate changed to the transmitted update. If $NACK$ is received, the transmitter knows that the receiver did not receive the new update, and the receiver's estimate did not change. Hence, we can assume that the transmitter always knows the receiver's estimate.

\subsection{Age of Incorrect Information}\label{sec-AoII}
We consider the Age of Incorrect Information (AoII), $\Delta_{AoII}(X_t,\hat{X}_t,t)$, where the age increases as long as the receiver is unaware of the correction information of the source process and the increment of age is enhanced by the mismatch between $X_t$ and $\hat{X}_t$. We define $U_t$ as the last time instant before time $t$ (including $t$) that the receiver's estimate is correct. Then, $\Delta_{AoII}(X_t,\hat{X}_t,t)$ can be written as
\begin{equation}\label{eq-AoII}
\Delta_{AoII}(X_t,\hat{X}_t,t) = \sum_{s=U_t+1}^{t}\left(g(X_s,\hat{X}_s)\times F(s-U_t)\right),
\end{equation}
where $g(X_t,\hat{X}_t)$ can be any function that reflects the mismatch between $X_t$ and $\hat{X}_t$. $F(t) \triangleq f(t) - f(t-1)$ where $f(t)$ can be any non-decreasing time function. In this paper, we let $X_t\in\{1,2,...,N\}$ be the state of the source process, $g(X_t,\hat{X}_t)=|X_t-\hat{X}_t|$, and $f(t) = t$. Consequently, $F(t)=1$ by its definition and $d_t\triangleq g(X_t,\hat{X}_t)\in\{0,1,...,N-1\}$. A sample path of $\Delta_{AoII}(X_t,\hat{X}_t,t)$ is shown in Fig. \ref{fig-SamplePath}. 

\begin{figure}
\centering
\includegraphics[width=\columnwidth]{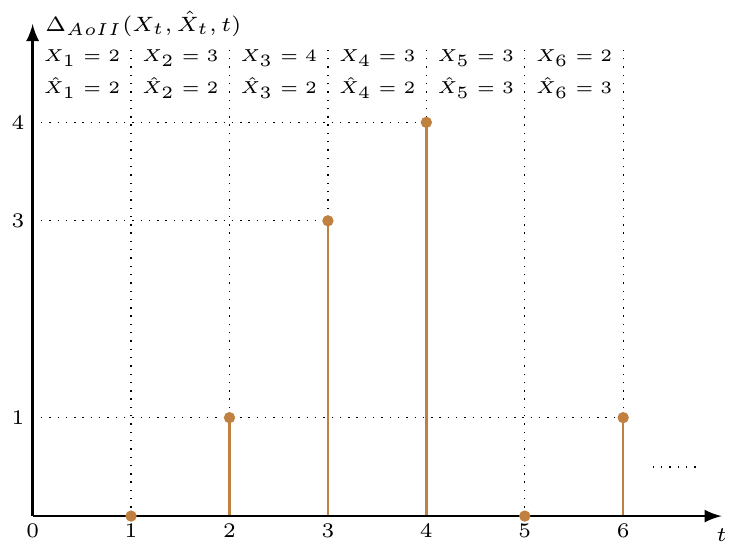}
\caption{A sample path of $\Delta_{AoII}(X_t,\hat{X}_t,t)$.}
\label{fig-SamplePath}
\end{figure}

\subsection{System Dynamic}\label{sec-SystemDynamic}
Now, we tackle down the system dynamic which can be fully captured by the dynamic of the pair $(d_t,\Delta_t)$. $\Delta_t$ is short for $\Delta_{AoII}(X_t,\hat{X}_t,t)$. Thus, it is essential to characterize the relationship between $(d_{t+1},\Delta_{t+1})$ and $(d_t,\Delta_t)$. We notice that the relationship depends on the transmitter's action and its result. Therefore, we define $a_t\in\{0,1\}$ as the transmitter's action at time $t$ where $a_t=1$ if the transmitter makes the transmission attempt and $a_t=0$ otherwise. Then, we can divide our discussion into the following three cases.
\paragraph{Case 1} $a_t=0$. In this case, no new update is received by the receiver. Thus, the estimate $\hat{X}_{t+1}$ will be nothing but $\hat{X}_t$, and $X_t$ will evolve following the Markov chain shown in Fig. \ref{fig-MarkovianSource}. When the state of the source process does not change which happens with probability $1-2p$, we have $X_{t+1}=X_t$. Then, $d_{t+1}=d_t$. When the state of the source process changes, $d_{t+1}$ depends on the value of $d_t$. Thus, we further distinguish between the following cases:
\begin{itemize}
\item  When $d_t=0$, according to the Markovin source reported in Fig. \ref{fig-MarkovianSource}, $d_{t+1}=1$ with probability $2p$.
\item When $1\leq d_t\leq N-2$, to simplify our analysis, we assume, for any $X_t\in\{1,2,...,N\}$, $Pr(X_{t+1} = X_t-1\ |\ X_t)=Pr(X_{t+1} = X_t+1\ |\ X_t)=p$. Then, when $X_t>\hat{X}_t$, $d_{t+1}=|X_{t}\pm1-\hat{X}_t|=|d_t\pm1|=d_t\pm1$. When $X_t<\hat{X}_t$, $d_{t+1}=|X_{t}\pm1-\hat{X}_t|=|-d_t\pm1|=d_t\mp1$. Combining together, $d_{t+1} = d_t\pm1$ with equal probability $p$.
\item When $d_t=N-1$, $X_t$ must be either $1$ or $N$ and $\hat{X}_t$ must be either $N$ or $1$, respectively. Combining with the Markovin source reported in Fig. \ref{fig-MarkovianSource}, $d_{t+1} = N-2$ with probability $2p$.
\end{itemize}
Let us denote by $P_{d,d'}$ the transition probability from $d$ to $d'$. Then, the results can be summarized as follows.
\begin{equation}\label{eq-TransitionPro}
\begin{cases}
& P_{d,d} = 1-2p \quad \ \ \ \ \ \ \ \ for\ 0\leq d\leq N-1,\\
& P_{d,d+1} = P_{d,d-1} = p \quad for\ 1\leq d\leq N-2,\\
& P_{0,1} = P_{N-1,N-2} = 2p.
\end{cases}
\end{equation}
Such dynamic can be characterized by the Markov chain shown in Fig. \ref{fig-EvolutionD} with $p_1 = 1-2p$ and $p_2 = p$.

According to \eqref{eq-AoII}, the value of $\Delta_{t+1}$ can be captured by the following two cases.
\begin{itemize}
\item When $d_{t+1} = 0$, the receiver's estimate at time $t+1$ is correct. By definition, $U_{t+1} = t + 1$. Hence, $\Delta_{t+1}=0$.
\item When $d_{t+1}\neq0$, the receiver's estimate at time $t+1$ is incorrect. In this case, $U_{t+1} = U_t$ by definition. Therefore, $\Delta_{t+1} = \Delta_{t} + d_{t+1}$.
\end{itemize}
To sum up,
\[
\Delta_{t+1} = \mathbbm{1}_{\{d_{t+1}\neq0\}}\times (\Delta_t + d_{t+1}).
\]
\iffalse
\begin{itemize}
\item When $U_{t+1}\neq t+1$, the receiver's estimate at time $t+1$ is incorrect. Then, we know $d_{t+1}\neq0$. Therefore, we can verify that $\Delta_{t+1} = \Delta_{t} + d_{t+1}$ using \eqref{eq-AoII}.
\item When $U_{t+1} = t+1$, we have $d_{t+1}=0$ by the definition of $U_{t+1}$. Then, $\Delta_{t+1}=0$ according to \eqref{eq-AoII}.
\end{itemize}
To sum up,
\begin{equation}\label{eq-DynamicDelta}
\Delta_{t+1} = \mathbbm{1}_{\{U_{t+1}\neq t+1\}}\times (\Delta_t + d_{t+1}).
\end{equation}
\fi
\paragraph{Case 2} $a_t=1$ but $r_t=0$. In this case, we notice that no new update is received by the receiver. Following the same trajectory as in \textit{Case 1}, we can conclude that the dynamic of $d_t$ can be characterized by the Markov chain shown in Fig. \ref{fig-EvolutionD} with $p_1 = p_f(1-2p)$ and $p_2 = p_fp$. $\Delta_{t+1}$ is fully dictated by $\Delta_t$ and $d_{t+1}$ as detailed in \textit{Case 1}.
\paragraph{Case 3} $a_t = 1$ and $r_t=1$. In this case, the receiver receives the update instantly. Thus, $U_{t+1} = t$. Then, we can conclude that $\Delta_{t+1} = d_{t+1}$ using \eqref{eq-AoII}. Since the update is received instantly, we have $d_{t+1}\in\{0,1\}$. More precisely,
\iffalse
Meanwhile, for the receiver, there is no additional new update received at time $t+1$. Combining together, we can conclude that it is equivalent to \textit{Case 1} with $(d_t,\Delta_t) = (0,0)$. 
\fi
\begin{equation*}
\begin{split}
& Pr\big((d_{t+1},\Delta_{t+1}) = (0,0) | (d_t,\Delta_t), a_t=1\big) = p_s(1-2p),\\
& Pr\big((d_{t+1},\Delta_{t+1}) = (1,1) | (d_t,\Delta_t), a_t=1\big) = 2p_sp.
\end{split}
\end{equation*}
Combining the above three cases, we can fully capture the evolution of $(d_t, \Delta_t)$.
\subsection{Problem Formulation}
We consider the problem where there is a unit power consumption along with each transmission attempt regardless of the result. At the same time, the transmitter has a power budget $\alpha<1$. We define $\phi = (a_0, a_1, ...)$ as a sequence of actions the transmitter takes and denote all the feasible series of actions as $\Phi$. Then, our problem can be formulated as
\begin{argmini!}|l|
{\phi \in \Phi} {\bar{\Delta}_{\phi}\triangleq\lim_{T\to\infty} \frac{1}{T} \mathbb{E}_{\phi}\left[\sum_{t=0}^{T-1}\Delta_t \mid \mathcal{X}_0\right]}{\label{eq-Constrained}}{\label{eq-ObjectValue}}
\addConstraint{\bar{R}_{\phi}\triangleq \lim_{T\to\infty} \frac{1}{T} \mathbb{E}_{\phi}\left[\sum_{t=0}^{T-1} a_t \mid \mathcal{X}_0\right]\leq\alpha,}\label{eq-ObjectConstraint}
\end{argmini!}
where $\mathcal{X}_0$ is the initial state of the system. As \eqref{eq-ObjectConstraint} shows, the system is resource-constrained. Thus, we realize a necessity to require the transmission attempts to help minimize AoII. More precisely, we require $Pr((0,0)\mid (d_t,\Delta_t),a_t=1) \geq Pr((0,0)\mid (d_t,\Delta_t),a_t=0)$ for any $(d_t,\Delta_t)$. Leveraging the system dynamic in Section \ref{sec-SystemDynamic}, we conclude that it is sufficient to require $p\in [0,\frac{1}{3}]$. Therefore, we only consider the case of $p\in [0,\frac{1}{3}]$ throughout the rest of this paper. A similar assumption is also made in \cite{b12}.

We notice that solving problem \eqref{eq-Constrained} is equivalent to solving a \textit{Constrained Markov Decision Process (CMDP)}. To this end, we adopt the Lagrangian approach.

\section{Problem Optimization}\label{sec-ProblemOptimization}
\subsection{Lagrangian Approach}\label{sec-Lagrange}
First of all, we write the constrained optimization problem \eqref{eq-Constrained} into its Lagrangian form.
\[
\mathcal{L}(\phi, \lambda) = {\lim_{T\to\infty} \frac{1}{T} \mathbb{E}_{\phi}\left[\sum_{t=0}^{T-1}(\Delta_t+\lambda a_t) \ | \ \mathcal{X}_0\right]}-\lambda\alpha,
\]
where $\lambda$ is the Lagrange multiplier. Then, the corresponding dual function will be
\begin{equation}\label{eq-Dual}
\mathcal{G}(\lambda)=\underset{\phi \in \Phi}{\text{min}}\ \mathcal{L}(\phi, \lambda).
\end{equation}
According to the results in \cite{b15}, the optimal policy for the constrained problem \eqref{eq-Constrained} can be characterized by the optimal policies for the minimization problem \eqref{eq-Dual} under certain $\lambda$. Thus, we start with solving problem \eqref{eq-Dual} for any given $\lambda\geq0$. As $\lambda\alpha$ is independent of policy, we can ignore it which leads to the following optimization problem.
\begin{mini}|l|
{\phi \in \Phi} {\lim_{T\to\infty} \frac{1}{T} \mathbb{E}_{\phi}\left[\sum_{t=0}^{T-1}(\Delta_t +\lambda a_t)\ | \ \mathcal{X}_0\right].}{\label{eq-Minimization}}{}
\end{mini}
The above problem can be cast into an infinite horizon with average cost \textit{Markov Decision Process (MDP)} $\mathcal{M} = (\mathcal{X},\mathcal{A},\mathcal{P},\mathcal{C})$, where
\begin{itemize}
\item $\mathcal{X}$ denotes the state space: the state is $x=(d,\Delta)$. We define $x_d=d$ and $x_{\Delta}=\Delta$. We will use $x$ and $(d,\Delta)$ to represent the state interchangeably for the rest of this paper.
\item $\mathcal{A}$ denotes the action space: the two feasible actions are making the transmission attempt $(a=1)$ and staying idle $(a=0)$. The action space is independent of the state and the time. More precisely, $a\in\mathcal{A}(x,t)=\mathcal{A}=\{0,1\}$.
\item $\mathcal{P}$ denotes the state transition probabilities: we define $P_{x,x'}(a)$ as the probability that action $a$ at state $x$ will lead to state $x'$. The values of $P_{x,x'}(a)$ can be obtained easily from Section \ref{sec-SystemDynamic}.
\item $\mathcal{C}$ denotes the instant cost: when the system is at state $x$ and action $a$ is chosen, the instant cost is $C(x,a)=x_{\Delta}+\lambda a$.
\end{itemize}

\subsection{Structural Results}\label{sec-Structural}
In this section, we provide the key structural properties of the optimal policy for $\mathcal{M}$, which plays a vital role in the analysis later on. The optimal policy for $\mathcal{M}$ is captured by its value function $V(x)$, which can be obtained by solving the Bellman equation. In the infinite horizon with average cost \textit{MDP}, the Bellman equation is defined as
\begin{equation}\label{eq-Bellman}
\theta+V(x)= x_{\Delta} + \min_{a \in \{0,1\}} \left\{\lambda a+\sum_{x'\in\mathcal{X}}P_{x,x'}(a)V(x')\right\},
\end{equation}
where $\theta$ is the minimal value of \eqref{eq-Minimization}. A canonical procedure to solve the Bellman equation is applying the \textit{Relative Value Iteration (RVI)}. To this end, we denote by $V_{\nu}(\cdot)$ the estimated value function at iteration $\nu$ of \textit{RVI}. We initialize $V_0(x) = x_{\Delta}$. Then, the estimated value function is updated in the following way.
\begin{equation}\label{eq-BellmanRelative}
V_{\nu+1}(x) = Q_{\nu+1}(x) - Q_{\nu+1}(x^{ref}),
\end{equation}
where $x^{ref}$ is the reference state. $Q_{\nu+1}(x)$ is the interim value function and is calculated by applying the right-hand side of \eqref{eq-Bellman}. More precisely,
\begin{equation}\label{eq-BellmanUpdate}
Q_{\nu+1}(x)= x_{\Delta} + \min_{a \in \{0,1\}} \left\{\lambda a + \sum_{x'\in\mathcal{X}}P_{x,x'}(a)V_{\nu}(x')\right\}.
\end{equation}
\textit{RVI} is guaranteed to converge to $V(\cdot)$ when $\nu\rightarrow+\infty$ regardless of the initialization \cite{b16}. However, it requires infinitely many iterations to achieve the exact solution. To conquer the impracticality, we leverage the special properties of our system and provide the structural property of $V_{\nu}(\cdot)$, which turns out to be enough to characterize the structure properties of the optimal policy for $\mathcal{M}$. We start with the following lemma.
\begin{lemma}[Monotonicity]\label{le-IncreasingV}
The estimated value function $V_{\nu}(x)$ is increasing in both $x_d$ and $x_{\Delta}$ at any iteration $\nu$.
\end{lemma}
\begin{proof}
Leveraging the iterative nature of \textit{RVI}, we use induction to prove the desired results. The complete proof is in Appendix \ref{proof-IncreasingV}.
\end{proof}
We refer state $x$ as active if the optimal action at $x$ is making the transmission attempt and inactive otherwise. Then, leveraging Lemma \ref{le-IncreasingV}, we provide the key structural properties of the optimal policy for $\mathcal{M}$.
\begin{proposition}[Structural properties]\label{prop-ThresholdPolicy}
The optimal policy for $\mathcal{M}$ under any $\lambda\geq0$ is a threshold policy which possesses the following properties.
\begin{itemize}
\item State $(0,0)$ will never be active.
\item For states $x$ with fixed $x_d\neq0$, the optimal action $a^*$ will switch from $a^*=0$ to $a^*=1$ as $x_{\Delta}$ increases and the switching point (i.e. threshold) is non-increasing in $x_d$.
\end{itemize}
\end{proposition}
\begin{proof}
The optimal action at state $x$ is captured by the sign of $\delta V(x) \triangleq V^1(x) - V^0(x)$ where $V^a(x)$ is the value function resulting from taking action $a$ at state $x$. Then, we characterize the sign of $\delta V(x)$ using Lemma \ref{le-IncreasingV}. The complete proof is in Appendix \ref{proof-OptimalPolicy}.
\end{proof}
We define the threshold for the states with fixed $d\neq0$ as the smallest $\Delta$ such that $a^*=1$. We notice that state $(0,0)$ will never be active if optimal policy is adopted. Hence, we can characterize the optimal policy using the thresholds. In the following, an optimal policy is denoted by a vector $\bm{n}_{\lambda}$ where $(\bm{n}_{\lambda})_{i}$ is the threshold for the states with $d=i$. The subscript $\lambda$ indicates the dependency between the optimal policy and $\lambda$.

\subsection{Finite-State MDP Approximation}\label{sec-RVIA}
In the sequel, we tackle down the problem of finding the optimal policy for $\mathcal{M}$. Direct application of \textit{RVI} becomes impractical as we need to estimate infinitely many value functions at each iteration. To overcome this difficulty, we use \textit{Approximating Sequence Method (ASM)} \cite{b17} and rigorously show the convergence of this approximation. To this end, we construct another $\mathcal{M}^{(m)} = (\mathcal{X}^{(m)},\mathcal{A},\mathcal{P}^{(m)},\mathcal{C})$ by truncating the value of $\Delta$. More precisely, we impose
\begin{subequations}
    \begin{empheq}[left=\mathcal{X}^{(m)}:\empheqlbrace\,]{align}
      & x_d^{(m)} \in \{0,1,...,N-1\}, \\
      & x_\Delta^{(m)} \in \{0,1,...,m\},
    \end{empheq}
\end{subequations}
where $m$ is the predetermined maximal value of $\Delta$. The transition probabilities from $x\in\mathcal{X}^{(m)}$ to $z\in\mathcal{X}-\mathcal{X}^{(m)}$ (called \textit{excess probabilities}) are redistributed to the states $x'\in\mathcal{X}^{(m)}$ in the following way.
\begin{equation*}
P^{(m)}_{x,x'}(a) = P_{x,x'}(a) + \sum_{z\in\mathcal{X}-\mathcal{X}^{(m)}}P_{x,z}(a)q_{z}(x'),
\end{equation*}
where $q_{z}(x')$ is the probability of distributing state $z$ to state $x'$ and $\sum_{x'\in\mathcal{X}^{(m)}}q_{z}(x')=1$. We choose $q_{z}(x')=\mathbbm{1}_{\{x'_d=z_d\}}\times\mathbbm{1}_{\{x'_\Delta=m\}}$. So, the transition probabilities $P^{(m)}_{x,x'}(a)$ satisfy the following.
\begin{equation*}
P^{(m)}_{x,x'}(a) = \begin{cases}
          P_{x,x'}(a) &\text{if}\ \ x'_\Delta<m, \\
          P_{x,x'}(a) + \sum_{G(z,x')}P_{x,z}(a) &\text{if}\ \ x_\Delta'= m, \\
     \end{cases}
\end{equation*}
where $G(z,x') = \{z:z_d=x_d',z_\Delta>m\}$. The action space $\mathcal{A}$ and the instant cost $\mathcal{C}$ are the same as defined in $\mathcal{M}$.
\begin{theorem}[Convergence]\label{theo-ASMConvergence}
The sequence of optimal policies for $\mathcal{M}^{(m)}$ will converge to the optimal policy for $\mathcal{M}$ as $m\rightarrow\infty$.
\end{theorem}
\begin{proof}
We show that our system verifies the two assumptions given in \cite{b17}. Then, the convergence is guaranteed according to the results in \cite{b17}. The complete proof is in Appendix \ref{proof-ASMConvergence}.
\end{proof}
For a given truncation parameter $m$, the state space $\mathcal{X}^{(m)}$ is finite with size $|\mathcal{X}^{(m)}| \propto m$. When $m$ is huge, the basic \textit{RVI} will be inefficient since the minimum operator in \eqref{eq-BellmanUpdate} requires calculations for both feasible actions at every state. In the following, we propose an improved \textit{RVI} which avoids minimum operators at certain states. To this end, we claim that the properties in Proposition \ref{prop-ThresholdPolicy} are also possessed by the optimal policies for $\mathcal{M}^{(m)}$ at any iteration $\nu$ of \textit{RVI}. The proof is omitted since it is very similar to what we did in Section \ref{sec-Structural}. Utilizing Proposition \ref{prop-ThresholdPolicy}, we can conclude that, for any state $x^{(m)}$, if there exists an active state $y^{(m)}$ such that $y^{(m)}_{\Delta}\leq x^{(m)}_{\Delta}$ and $y^{(m)}_d\leq x^{(m)}_d$, then $x^{(m)}$ must also be active. The update step at each iteration of the improved \textit{RVI} can be summarized as follows. For each $x^{(m)}\in\mathcal{X}^{(m)}$,
\begin{itemize}
\item if $y^{(m)}$ exists, we can determine the optimal action at $x^{(m)}$ immediately, and the minimum operator is avoided.
\item if $y^{(m)}$ does not exist, the optimal action at $x^{(m)}$ is determined by applying the minimum operator.
\end{itemize}
In this way, we avoid $\sum_{i=1}^N(m-n_i)$ minimum operators at each iteration of \textit{RVI} where $n_i=(\bm{n}^t)_i$ and $\bm{n}^t$ is the optimal policy at iteration $\nu$ of \textit{RVI}. In almost all cases, we have $n_i\ll m$. The pseudocode is given in Algorithm \ref{alg-RVIA} of Appendix \ref{sec-Algorithm}.

\iffalse
In order to get a close enough approximation of the optimal policy, we require $m$ to be large enough such that 
\begin{equation}\label{eq-choosem}
\max_{\theta\in\theta^{(m)}}\left\lbrace P_{m,\theta}\right\rbrace\leq\epsilon_\alpha
\end{equation}
where $\theta^{(m)}$ is the set of all possible policies in $\mathcal{M}^{(m)}$, $P_{m,\theta}$ is the probability of the system being at state $x\in\mathcal{M}-\mathcal{M}^{(m)}$ under the policy $\theta$ and $\epsilon_{\alpha}$ is a predetermined threshold depending on $\alpha$. We choose $\epsilon_{\alpha}=\frac{\alpha}{2}$ in this problem. Since the policy can be fully characterized by the threshold vector $\bm{n}$, we know $P_{m,\theta}=P_{m,\bm{n}}$. Moreover, $P_{m,\bm{n}}$ is "decreasing" in the threshold values. Combining with the fact that we are searching among all the possible policies in $\mathcal{M}^{(m)}$, we can conclude that $\max_{\theta\in\theta^{(m)}}\left\lbrace P_{m,\theta}\right\rbrace=P_{m,\bm{n}_m}$ where $\bm{n}_m=[m,m,...,m]$. Consequently, we start from $m=100$, calculate $P_{m,\bm{n}_m}$ and increase $m$ until \eqref{eq-choosem} is satisfied.
\fi

\subsection{Expected Transmission Rate}\label{sec-TransmissionRate}
In this section, we calculate the expected transmission rate $\bar{R}_{\bm{n}}$ under given threshold policy $\bm{n}$. It enables us to develop an efficient algorithm for finding the optimal policy for \eqref{eq-Constrained}. As we can see in \eqref{eq-ObjectConstraint}, $\bar{R}_{\bm{n}}$ is nothing but the expected average number of transmission attempts made. Thus, it can be fully captured by the stationary distribution of the \textit{Discrete-Time Markov Chain (DTMC)} induced by $\bm{n}$. More precisely,
\begin{equation*}
\bar{R}_{\bm{n}} = \sum_{d=1}^{N-1}\sum_{\Delta=n_d}^{+\infty}\pi_d(\Delta),
\end{equation*} 
where $\pi_d(\Delta)$ is the steady-state probability of state $(d,\Delta)$ and $n_d = (\bm{n})_d$. To obtain the stationary distribution, we utilize the balance equation associated with the induced \textit{DTMC} which takes the following form
\begin{equation}\label{eq-BalanceEquation}
\pi(x) = \sum_{x'\in\mathcal{X}}P_{x',x}(a^{\bm{n}}_{x'})\pi(x'),
\end{equation}
where $a^{\bm{n}}_{x'}$ is the action suggested by policy $\bm{n}$ at state $x'$. The problem arises since the state space of the induced \textit{DTMC} is infinite. To overcome the difficulty, we present the following proposition.
\begin{proposition}[Expected transmission rate]\label{prop-TransmissionRate}
The expected transmission rate under threshold policy $\bm{n}$ is
\[
\bar{R}_{\bm{n}} = \sum_{d=1}^{N-1}\left(\sum_{\Delta=n_d}^{\tau-1}\pi_d(\Delta) + \Pi_d(\tau)\right),
\]
where $\tau=max\{\bm{n}\}$. $\pi_d(\Delta)$'s and $\Pi_d(\tau)$'s are the solution to the following finite system of linear equations.\\
For each $1\leq d\leq N-1$:
\begin{equation}\label{eq-BalanceGeneral}
\begin{cases}
& \pi_d(\Delta) = 0\ \ \ \ \ \ \ \ \ \ \ \ \ \ \ \ \ \ \ \ \ \ \ \ \ \ \ \ \ \ \ \ \ \ \ \ \ \ \ \ \ \ \ \ \ \ \ \ for\ 0<\Delta<l_d,\\
&
\begin{split}
\pi_d(\Delta) &= \sum_{d'=1}^{N-1}P_{d',d}(1-p_sa_{d',\Delta-d})\pi_{d'}(\Delta-d)\ \ \ \ \ \ \ \ for\ \max\{2,l_d\}\leq\Delta\leq \tau-1,
\end{split}\\
&
\begin{split}
\Pi_d(\tau) & = \sum_{d'=1}^{N-1}P_{d',d}\Bigg(\sum_{\Delta=\tau-d}^{\tau-1}(1-p_sa_{d',\Delta})\pi_{d'}(\Delta) + p_f\Pi_{d'}(\tau)\Bigg),
\end{split}\\
\end{cases}
\end{equation}
\begin{equation}\label{eq-Balance0}
\begin{split}
\pi_0(0) =& (1-2p)\pi_0(0) + p\sum_{\Delta=1}^{\tau-1}(1-p_sa_{1,\Delta})\pi_1(\Delta) + p_fp\Pi_1(\tau)+\\
& p_s(1-2p)\sum_{d=1}^{N-1}\left(\sum_{\Delta=n_d}^{\tau-1}\pi_d(\Delta) + \Pi_d(\tau)\right),
\end{split}
\end{equation}
\begin{equation}\label{eq-Balance1}
\begin{split}
\pi_1(1) = 2p\pi_0(0)+2p_sp\sum_{d=1}^{N-1}\left(\sum_{\Delta=n_d}^{\tau-1}\pi_d(\Delta)+\Pi_d(\tau)\right),
\end{split}
\end{equation}
\begin{equation}\label{eq-BalanceSum}
\begin{split}
\sum_{d=1}^{N-1}\left(\sum_{\Delta=l_d}^{\tau-1}\pi_d(\Delta) + \Pi_d(\tau)\right) + \pi_0(0) = 1,
\end{split}
\end{equation}
where $l_d=\frac{d^2+d}{2}$ and $a_{d,\Delta}$ is the action suggested by $\bm{n}$ at state $(d,\Delta)$.
\end{proposition}
\begin{proof}
We cast the induced infinite-state Markov chain to an equivalent finite-state Markov chain with size depending on the policy. Then, $\pi_d(\Delta)$'s and $\Pi_d(\tau)$'s are the steady-state probabilities of the finite-state Markov chain. The complete proof is in Appendix \ref{proof-TransmissionRate}.
\end{proof}
\begin{remark}\label{rmk-overdetermined}
We can verify that \eqref{eq-Balance0} is a linear combination of the other equations in  the system of linear equations. Therefore, we can exclude \eqref{eq-Balance0} in practice.
\end{remark}
The system of linear equations can be reformulated into the matrix form $\bm{A}\bm{\pi} = \bm{b}$ where $\bm{\pi}$ is the unknowns and $\bm{A}$, $\bm{b}$ can be obtained easily from Proposition \ref{prop-TransmissionRate}. However, solving a finite system of linear equations can still be problematic, especially when the system is huge. For our problem, the size of $\bm{A}$ is $\mathcal{O}((N-1)\tau)$, and we notice that $\bm{A}$ is sparse.

In general, solving a large system of linear equations of size $\mathcal{O}(n)$ requires $\mathcal{O}(n^2)$ storage and $\mathcal{O}(n^3)$ floating-point arithmetic operations when $\bm{A}$ is dense. In the case of sparse $\bm{A}$, the computational cost will be less. The sparse matrix algorithms are designed to solve equations in time and space proportional to $\mathcal{O}(n) + \mathcal{O}(cn)$ where $c$ is the average number of non-zero entries in each column. Although there are cases where this linear target cannot be met, the complexity of sparse linear algebra is far less than that in dense case \cite{b18}. Generally speaking, the complexity depends on the sparsity of $\bm{A}$. By exploiting the zero entries, we can often reduce the storage and computational requirements to $\mathcal{O}(cn)$ and $\mathcal{O}(cn^2)$, respectively.

The calculation of $\bar{R}_{\bm{n}}$ is constantly needed, and it requires a significant amount of time when the thresholds in $\bm{n}$ are huge. Hence, we provide an efficient alternative that can approximate the expected transmission rate in this case.
\begin{corollary}[Approximation]\label{prop-RateApproximation}
When the thresholds in $\bm{n}$ are huge, the expected transmission rate under policy $\bm{n}$ can be approximated as
\[
\bar{R}_{\bm{n}} \approx \sum_{d=1}^{N-1}\left(\sum_{\Delta=n_d}^{\tau-1}\pi_d(\Delta) + \Pi_d(\tau)\right),
\]
where $\tau=max\{\bm{n}\}$. $\pi_d(\Delta)$'s and $\Pi_d(\tau)$'s are the solution to the following finite system of linear equations.\\
For each $1\leq d\leq N-1$:
\begin{equation}\label{eq-ApproximationD}
\begin{cases}
&
\begin{split}
\pi_d(\Delta) &= \rho\sum_{d'=1}^{N-1}P_{d',d}(1-p_sa_{d',\Delta-d})\sigma_{d'}(\Delta-d)\ \ \ \ \ \ \ \ for\ \eta+1\leq\Delta\leq\eta+d,
\end{split}\\
&
\begin{split}
\pi_d(\Delta) &= \sum_{d'=1}^{N-1}P_{d',d}(1-p_sa_{d',\Delta-d})\pi_{d'}(\Delta-d)\ \ \ \ \ \ \ \ \ \ for\ \eta+d+1\leq\Delta\leq\tau-1,
\end{split}\\
&
\begin{split}
\Pi_d(\tau) & = \sum_{d'=1}^{N-1}P_{d',d}\Bigg(\sum_{\Delta=\tau-d}^{\tau-1}(1-p_sa_{d',\Delta})\pi_{d'}(\Delta)+ p_f\Pi_{d'}(\tau)\Bigg),
\end{split}
\end{cases}
\end{equation}
\begin{equation}\label{eq-Approximation0}
\begin{split}
\pi_0(0) = & (1-2p)\pi_0(0) + p\Pi_1(\eta) + p_fp\Pi_1(\tau) +  p\sum_{\Delta=\eta+1}^{\tau-1}(1-p_sa_{1,\Delta})\pi_1(\Delta)+\\
& p_s(1-2p)\sum_{d=1}^{N-1}\left(\sum_{\Delta=n_d}^{\tau-1}\pi_d(\Delta)+\Pi_d(\tau)\right),
\end{split}
\end{equation}
\begin{equation}\label{eq-Approximation1}
\pi_1(1) = 2p\pi_0(0)+2p_sp\sum_{d=1}^{N-1}\left(\sum_{\Delta=n_d}^{\tau-1}\pi_d(\Delta) + \Pi_d(\tau)\right),
\end{equation}
\begin{equation}\label{eq-ApproximationSum}
\begin{split}
\sum_{d=1}^{N-1}\left(\Pi_d(\eta) + \sum_{\Delta=\eta+1}^{\tau-1}\pi_d(\Delta) + \Pi_d(\tau)\right) + \pi_0(0) = 1,
\end{split}
\end{equation}
\begin{equation}\label{eq-ApproximationHeads}
\begin{cases}
&
\begin{split}
\Pi_1(\eta) -\pi_1(1) + \pi_1(\eta+1) = \sum_{d'=1}^{N-1}P_{d',1}\Pi_{d'}(\eta),
\end{split}\\
&
\begin{split}
\Pi_d(\eta) + \sum_{\Delta=\eta+1}^{\eta+d}\pi_d(\Delta) &= \sum_{d'=1}^{N-1}P_{d',d}\Pi_{d'}(\eta)\ \ \ \ \ \ \ \ \ \ \ for\ 2\leq d\leq N-1,
\end{split}\\
\end{cases}
\end{equation}
where $a_{d,\Delta}$ is the action suggested by the threshold policy $\bm{n}$ at state $(d,\Delta)$, $\eta=\min\{\bm{n}\}-1$ and $\rho=\frac{\pi_0(0)}{\sigma_0(0)}$. $\sigma_d(\Delta)$ is the stationary distribution associated with the Markov chain induced by another threshold policy $\bm{n}' = [\eta',...,\eta']$ where $\eta'=\eta+1$.
\end{corollary}
\begin{proof}
When the thresholds in $\bm{n}$ are huge, the expected transmission rate will be insignificant. Combining with \eqref{eq-Balance1}, we have $\pi_1(1) \approx 2p\pi_0(0)$. Consequently, we can show that, for any state $(d,\Delta)$, $\pi_d(\Delta)\approx c^{\bm{n}}_{d,\Delta}\pi_0(0)$ where $c^{\bm{n}}_{d,\Delta}$ is a scalar that depends on the policy and the state. At the same time, we notice that, for any two threshold policies $\bm{n}_1$ and $\bm{n}_2$, the suggested actions at states with $\Delta<\min\{[\bm{n}_1,\bm{n}_2]\}$ are the same. We denote by $G(\bm{n}_1,\bm{n}_2)$ the set of these states. Then, we can prove that, for $(d,\Delta)\in G(\bm{n}_1,\bm{n}_2)$,
\begin{equation}\label{eq-Approx}
\frac{\pi^1_d(\Delta)}{\pi^2_d(\Delta)}\approx\frac{c_{d,\Delta}\pi^1_0(0)}{c_{d,\Delta}\pi^2_0(0)} = \frac{\pi^1_0(0)}{\pi^2_0(0)},
\end{equation}
where $\pi^1_d(\Delta)$ and $\pi^2_d(\Delta)$ are the stationary distributions when $\bm{n}_1$ and $\bm{n}_2$ are adopted, respectively. Based on \eqref{eq-Approx}, we can obtain the two systems of linear equations. The complete proof is in Appendix \ref{proof-RateApproximation}.
\end{proof}
\begin{remark}
For the same reason as in Remark \ref{rmk-overdetermined}, we can exclude \eqref{eq-Approximation0} in practice.
\end{remark}

In Corollary \ref{prop-RateApproximation}, instead of solving a large system of linear equations of size $\mathcal{O}((N-1)\tau)$, we approximate $\bar{R}_{\lambda}$ by solving two systems of linear equations of size $\mathcal{O}((N-1)(\eta+1))$ and $\mathcal{O}((N-1)(\tau-\eta+1))$, respectively. It is worth noting that when $\tau\approx\eta$ or $\tau\gg\eta$, the complexity reduction of Corollary \ref{prop-RateApproximation} is limited. For other cases, Corollary \ref{prop-RateApproximation} can significantly reduce the complexity and the resulting error is negligible. The methodology presented in Proposition \ref{prop-TransmissionRate} can also be applied to the calculation of the expected AoII $\bar{\Delta}_{\bm{n}}$. More precisely, we can use the following corollary.
\begin{corollary}[Expected AoII]\label{prop-ExpectedAoII}
The expected AoII under threshold policy $\bm{n}$ is
\[
\bar{\Delta}_{\bm{n}} = \sum_{d=1}^{N-1}\left(\sum_{\Delta=l_d}^{\tau-1}\omega_d(\Delta) + \Omega_d(\tau)\right),
\]
where $\tau=max\{\bm{n}\}$ and $l_d=\frac{d^2+d}{2}$. $\omega_d(\Delta)$'s and $\Omega_d(\tau)$'s are the solution to the following finite system of linear equations.\\
For each $1\leq d\leq N-1$:
\begin{equation}\label{eq-ExpectedAoIIGeneral}
\begin{cases}
& \omega_0(0) = 0; \omega_d(\Delta) = 0\ \ \ \ \ \ \ \ \ \ \ \ for\ \Delta<l_d,\\
&
\begin{split} 
\omega_d(\Delta) = \Delta\pi_d(\Delta) \ \ \ \ \ \ \ \ \ \ \ \ \ \ \ \ for\ l_d\leq\Delta\leq\tau-1,
\end{split}\\
&
\begin{split}
\Omega_d(\tau) & = \sum_{d'=1}^{N-1}P_{d',d}\Bigg(\sum_{\Delta=\tau-d}^{\tau-1}(1-p_sa_{d',\Delta})\omega_{d'}(\Delta) + p_f\Omega_{d'}(\tau)\Bigg) + d\Pi_d(\tau),\\
\end{split}\\
\end{cases}
\end{equation}
where $a_{d,\Delta}$ is the action suggested by the threshold policy $\bm{n}$ at state $(d,\Delta)$. $\pi_d(\Delta)$'s and $\Pi_d(\tau)$'s can be obtained using Proposition \ref{prop-TransmissionRate} with the same threshold policy $\bm{n}$.
\end{corollary}
\begin{proof}
We define $\omega_d(\Delta) \triangleq \Delta\pi_d(\Delta)$ and $\Omega_d(\tau) \triangleq \sum_{\Delta=\tau}^{+\infty}\omega_d(\Delta).$
Then, we combine the states with $\Delta\geq\tau$ as did in Proposition \ref{prop-TransmissionRate}. After some rearrangements, we can obtain the system of linear equations shown above. The complete proof is in Appendix \ref{proof-ExpectedAoII}.
\end{proof}

\subsection{Optimal Policy}\label{sec-OptimalPolicy}
Till this point, we are able to find the optimal policy for problem \eqref{eq-Minimization}. However, our goal is to find the optimal policy for the constrained problem \eqref{eq-Constrained}. Based on the work in \cite{b15}, the optimal policy for problem \eqref{eq-Constrained} can be expressed as a mixture of two deterministic policies that are both optimal for problem \eqref{eq-Minimization} with $\lambda=\lambda^*$. More precisely, the optimal policy can be summarized in the following theorem.
\begin{theorem}[Optimal policy]\label{theo-Construct}
The optimal policy for the constrained problem \eqref{eq-Constrained} can be expressed as a mixture of two deterministic policies $\bm{n}_{\lambda^*_+}$ and $\bm{n}_{\lambda^*_-}$ that are both optimal for problem \eqref{eq-Minimization} with $\lambda = \lambda^*\triangleq\inf\{\lambda>0:\bar{R}_{\lambda}\leq\alpha\}$. $\bar{R}_{\lambda}$ is the expected transmission rate resulting from policy $\bm{n}_{\lambda}$. More precisely, if we choose
\begin{equation}\label{eq-Mu}
\mu = \frac{\alpha - \bar{R}_{\lambda^*_+}}{\bar{R}_{\lambda^*_-}-\bar{R}_{\lambda^*_+}},
\end{equation}
the mixed policy $\bm{n}_{\lambda^*}$, which selects $\bm{n}_{\lambda^*_-}$ with probability $\mu$ and $\bm{n}_{\lambda^*_+}$ with probability $1-\mu$ each time the system reaches state $(0,0)$, is optimal for the constrained problem \eqref{eq-Constrained} and the constraint in \eqref{eq-ObjectConstraint} is met with equality.
\end{theorem}
\begin{proof}
We verify that our system satisfies all the assumptions given in \cite{b15}. Then, combining the characteristics of our system and the results in \cite{b15}, we obtain the optimal policy. The complete proof is in Appendix \ref{proof-Construct}.
\end{proof}
Next, we describe an efficient algorithm to obtain the optimal policy for the constrained problem \eqref{eq-Constrained}. The core of obtaining the optimal policy is to find $\lambda^*$. We recall that, for any given $\lambda$, the deterministic policy $\bm{n}_{\lambda}$ is obtained by applying the improved \textit{RVI} and the resulting $\bar{R}_{\lambda}$, which is non-increasing in $\lambda$ \cite{b15}, is calculated using Proposition \ref{prop-TransmissionRate}. Hence, $\bar{R}_{\lambda}$ can be regarded as a non-increasing function of $\lambda$, and we can use \textit{Bisection search} with tolerance $\xi$ to find $\lambda^*$. Then, $\lambda^*_+$ and $\lambda^*_-$ can be the boundaries of the final interval. More precisely, we initialize $\lambda_-=0$ and $\lambda_+=1$. Then, the procedure can be summarized as follows.
\begin{itemize}
\item As long as $\bar{R}_{\lambda_+}\geq\alpha$, we set $\lambda_-=\lambda_+$ and $\lambda_+=2\lambda_+$. Then, we end up with an interval $I = [\lambda_-,\lambda_+]$.
\item We apply \textit{Bisection Search} on the interval $I$ until the length of $I$ is less than the tolerance $\xi$. Then, the algorithm returns $\lambda^*_+$ and $\lambda^*_-$.
\end{itemize}
The pseudocode is given in Algorithm \ref{alg-BisectionSearch} of Appendix \ref{sec-Algorithm}. Finally, the mixing coefficient $\mu$ is calculated using \eqref{eq-Mu} and the resulting expected AoII is calculated using Corollary \ref{prop-ExpectedAoII}. The algorithm is efficient for the following reasons.
\begin{itemize}
\item We obtain $\bm{n}_{\lambda}$ using the improved \textit{RVI} which avoids minimum operators at certain states.
\item When calculating $\bar{R}_{\lambda}$, we cast the induced infinite-state Markov chain to a finite-state Markov chain. 
\item We find $\lambda^*_+$ and $\lambda^*_-$ using \textit{Bisection search} which has a logarithmic complexity.
\end{itemize}

\section{Numerical Results}\label{sec-NumericalResults}
In this section, we provide numerical results that accent the effect of system parameters on the performance of AoII-optimal policy. We also compare the AoII-optimal policy with the AoI-optimal policy derived in \cite{b21}.
\iffalse
\begin{figure*}%
\centering
\begin{subfigure}{.66\columnwidth}
\includegraphics[width=\columnwidth]{Figure/AoII-P.pdf}%
\caption{The expected AoII in function of $p$.}%
\label{fig-AoIIP}%
\end{subfigure}\hfill%
\begin{subfigure}{.66\columnwidth}
\includegraphics[width=\columnwidth]{Figure/AoII-Ps.pdf}%
\caption{The expected AoII in function of $p_s$.}%
\label{fig-AoIIPs}%
\end{subfigure}\hfill%
\begin{subfigure}{.66\columnwidth}
\includegraphics[width=\columnwidth]{Figure/AoII-Alpha.pdf}%
\caption{The expected AoII in function of $\alpha$.}%
\label{fig-AoIIAlpha}%
\end{subfigure}%
\caption{Illustrations of AoII-optimal policy. The truncation parameter $m=800$, the terminate criteria of \textit{RVI} $\epsilon=0.01$ and the convergence criteria of \textit{Bisection Search} $\xi=0.01$.}
\end{figure*}
\fi
\begin{figure*}%
\centering
\begin{subfigure}{.5\columnwidth}
\includegraphics[width=\columnwidth]{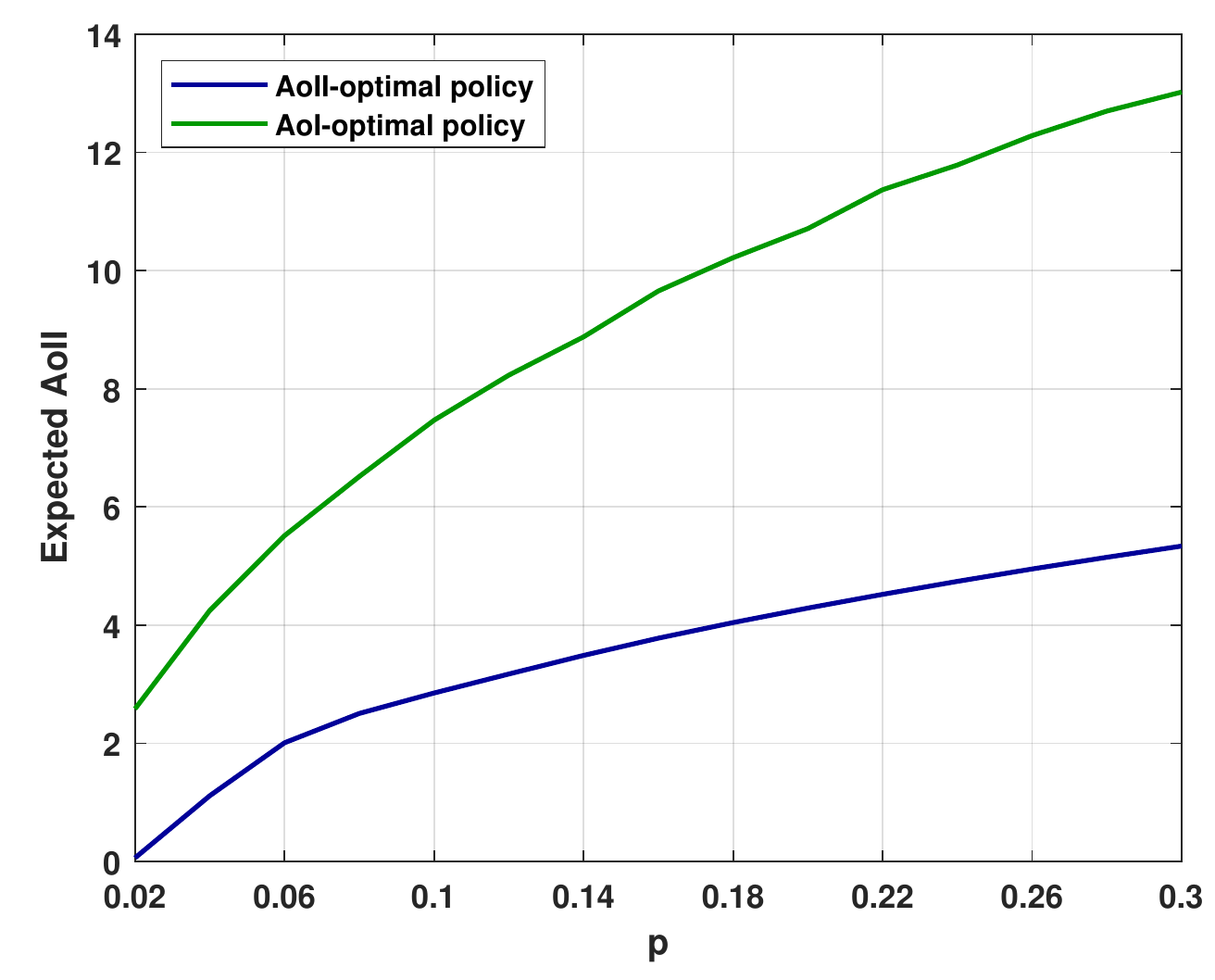}%
\caption{The expected AoII in function of $p$.}%
\label{fig-AoIIP}%
\end{subfigure}\hfill%
\begin{subfigure}{.5\columnwidth}
\includegraphics[width=\columnwidth]{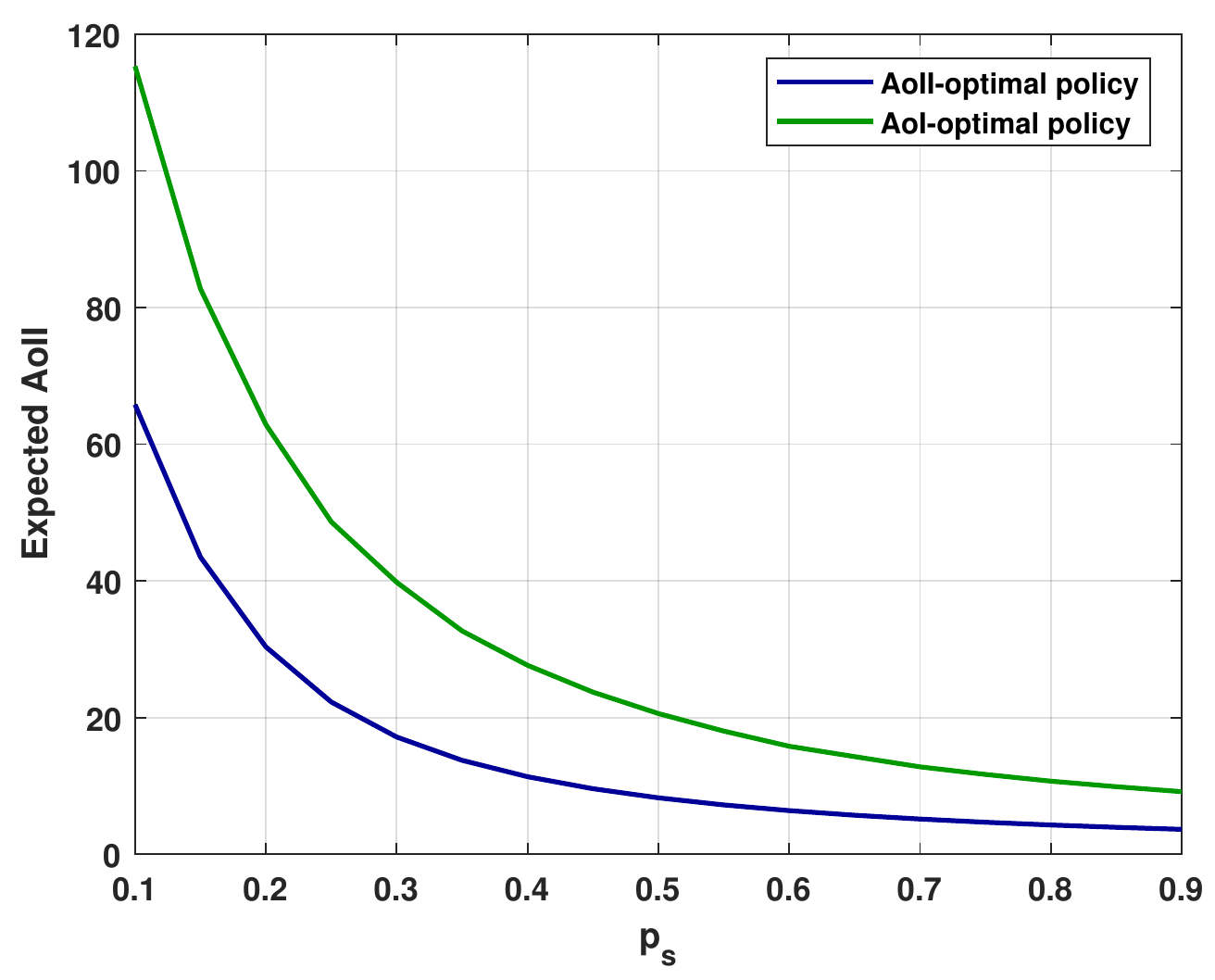}%
\caption{The expected AoII in function of $p_s$.}%
\label{fig-AoIIPs}%
\end{subfigure}\hfill%
\begin{subfigure}{.5\columnwidth}
\includegraphics[width=\columnwidth]{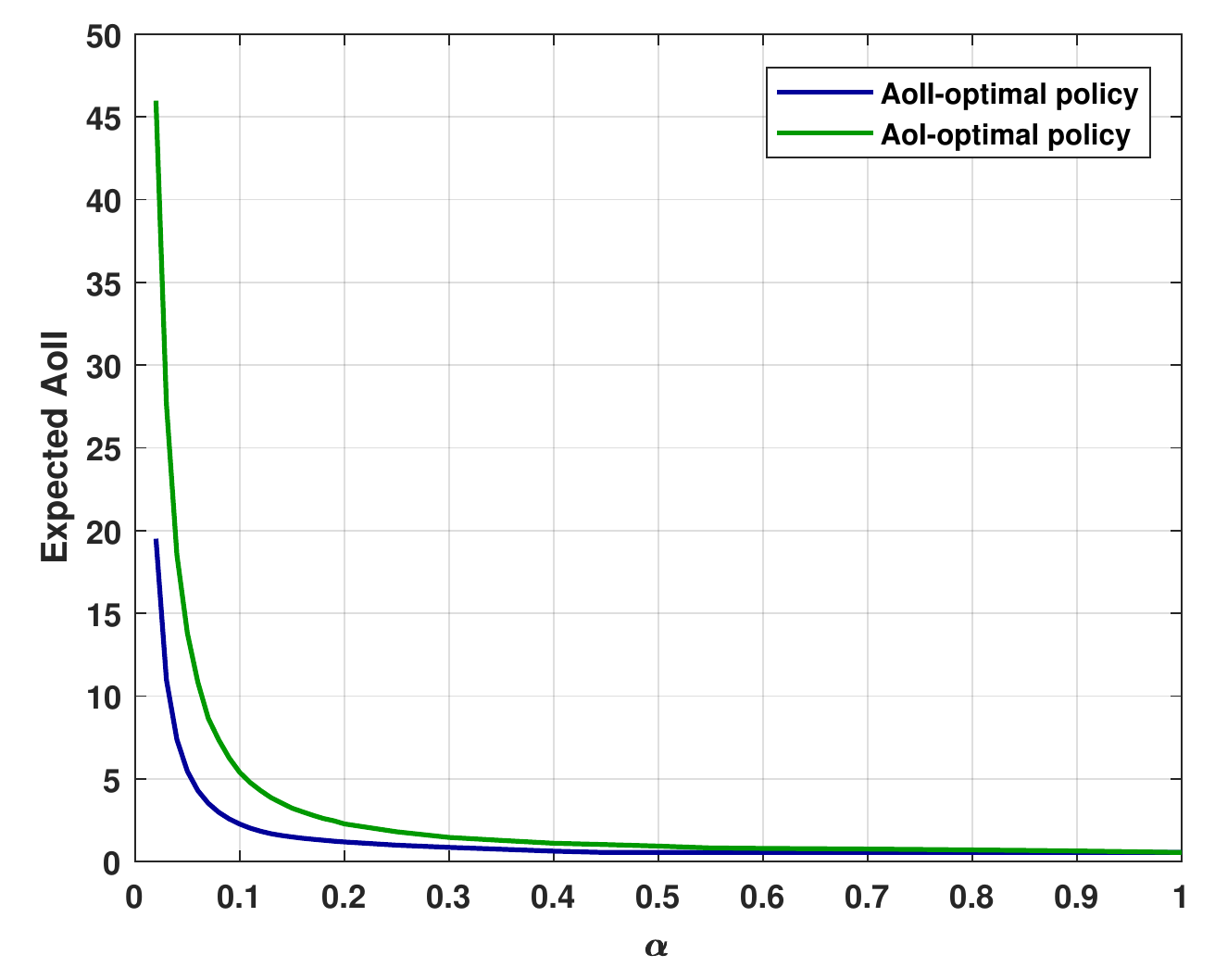}%
\caption{The expected AoII in function of $\alpha$.}%
\label{fig-AoIIAlpha}%
\end{subfigure}%
\caption{Illustrations of AoII-optimal policy and AoI-optimal policy. The truncation parameter in \textit{ASM} $m=800$ and the tolerance in \textit{Bisection search} $\xi=0.01$. \textit{RVI} converges when the maximum difference between the results of two consecutive iterations is less than $\epsilon = 0.01$.}
\label{fig-AoIIPerformance}
\end{figure*}

\paragraph{Effect of $p$}
We compare the performances of AoII-optimal policies under different values of $p$. To this end, we fix $N=7$ and $p_s=0.8$. We also set $\alpha=0.06$. We vary the value of $p$ and plot the corresponding results. As we can see in Fig. \ref{fig-AoIIP}, the expected AoII is increasing in $p$. To explain this trend, we notice that as $p$ increases, the source process will be more inclined to change state at the next time slot. Then, those successfully transmitted updates will more likely be obsolete at the next time slot. As the power budget $\alpha$ is fixed which dictates the transmission rate, the expected AoII will increase as $p$ increases.

We also show, in Table \ref{tab-PolicyP}, the deterministic policies $\bm{n}_{\lambda_+^*}$, $\bm{n}_{\lambda_-^*}$ and the corresponding mixing coefficient $\mu$ for some values of $p$.\footnote{In the table, " / " indicates the threshold where the two policies differ.} The optimal policy is the mixture of $\bm{n}_{\lambda_+^*}$ and $\bm{n}_{\lambda_-^*}$ with mixing coefficient $\mu$ as described in Theorem \ref{theo-Construct}. We can see that the thresholds are, in general, increasing in $p$. The reason behind this is as follows. When $p$ is small, the successfully transmitted updates are more likely still accurate in the next few time slots. In another word, the transmission is more "efficient". We refer a transmission as "efficient" if it reduces the age to the greatest extent. This allows the transmitter to make transmission attempts when the age is relatively low without violating the power constraint. 
\begin{table}[htbp]
\renewcommand{\arraystretch}{1.2}
\centering
\caption{Optimal thresholds for different $p$}
\begin{tabular}{c|c|c|c|c|c|c|c}
    \hline
    & Mixing Coef. & $n_1$ & $n_2$ & $n_3$ & $n_4$ & $n_5$ & $n_6$\\
    \hline
    \hline
    $p=0.1$ & $\mu=0.7176$ & 15 & 6/7 & 1 & 1 & 1 & 1\\ 
    \hline
    $p=0.2$ & $\mu=0.0331$ & 37 & 16 & 8/9 & 1 & 1 & 1\\ 
    \hline
    $p=0.3$ & $\mu=0.1178$ & 69 & 25/26 & 15 & 1 & 1 & 1\\ 
    \hline
\end{tabular}
\label{tab-PolicyP}
\end{table}

\paragraph{Effect of $p_s$}
In this scenario, we fix $p=0.2$ and investigate the effect of channel reliability $p_s$ on the performance of AoII-optimal policy. We still consider the case of $N=7$ and $\alpha=0.06$. The corresponding results are shown in Fig. \ref{fig-AoIIPs}. As $p_s$ increases, the expected AoII will decrease. The reason is as follows. As $p_s$ increases, the transmitted updates will more likely be successful. Consequently, the transmission will be more "efficient". As the power budget is fixed, the expected AoII will decrease as $p_s$ increases. We also present some selected thresholds in Table \ref{tab-PolicyPs}.
\begin{table}[htbp]
\renewcommand{\arraystretch}{1.2}
\centering
\caption{Optimal thresholds for different $p_s$}
\begin{tabular}{c|c|c|c|c|c|c|c}
    \hline
    & Mixing Coef. & $n_1$ & $n_2$ & $n_3$ & $n_4$ & $n_5$ & $n_6$\\
    \hline
    \hline
    $p_s=0.2$ & $\mu=0.6712$ & 556 & 228 & 140 & 96 & 70/71 & 60\\ 
    \hline
    $p_s=0.4$ & $\mu=0.3260$ & 151 & 62 & 36/37 & 24 & 17 & 1\\ 
    \hline
    $p_s=0.6$ & $\mu=0.4089$ & 67 & 27/28 & 16 & 1 & 1 & 1\\ 
    \hline
    $p_s=0.8$ & $\mu=0.0331$ & 37 & 16 & 8/9 & 1 & 1 & 1\\ 
    \hline
\end{tabular}
\label{tab-PolicyPs}
\end{table}
As we see in the table, the thresholds are, in general, decreasing in $p_s$. We recall that as $p_s$ increases, the transmission will be more "efficient". Thus, the transmitter can make transmission attempts when the age is relatively low while keeping the transmission rate not exceeding the power budget.

\paragraph{Effect of $\alpha$}
Then, we analyze the performances of AoII-optimal policies under different values of $\alpha$. We adopt $N=7$ and $p_s=0.8$. We also set $p=0.2$. The expected AoII achieved by the AoII-optimal policies are plotted in Fig. \ref{fig-AoIIAlpha}. As we see, the expected AoII decreases as $\alpha$ increases. The reason is simple. As the power budget increases, more transmission attempts are allowed. We recall that we impose a transmission attempt to always help reduce the age. Keeping this in mind, we can conclude that the expected AoII is decreasing in $\alpha$. It is worth noting that as $\alpha$ increases, the expected AoII will stop decreasing before $\alpha=1$. To explain this, we recall that the transmitter will never make transmission attempt at state $(0,0)$ if an optimal policy is adopted. Thus, as $\alpha$ becomes large, the transmitter will have enough budget to make transmission attempts at any states other than $(0,0)$. Then, the transmission rate is saturated, and the expected AoII will not decrease further.

\paragraph{Comparison with AoI-optimal policy}
Lastly, we compare the AoII-optimal policy with the AoI-optimal policy. From Fig. \ref{fig-AoIIPerformance}, we can see that the performance gap expands with the increase in $p$ and the decrease in $p_s$ and $\alpha$. The reason behind it lies in the value of transmission attempts. As $p$ increases, the source process becomes more inclined to change states. Therefore, transmission attempts are more needed to bring correct information to the receiver. As $p_s$ decreases, the number of successful transmissions will decrease, and the AoII will build up faster. In this case, transmission attempts are more valuable because they will greatly reduce AoII once they succeed. As $\alpha$ decreases, transmission attempts will be more valuable as fewer attempts are allowed. Due to the different definitions of age, there are often cases where AoI is large, but AoII is small. In these cases, the AoI-optimal policy will waste valuable transmission attempts. Therefore, the increase in the value of transmission attempts will lead to an expansion of the performance gap.

\section{Conclusion}
In this paper, we consider a system where the source process is modeled by an N-state Markov chain. The AoII which considers the quantified mismatch between the source and the knowledge at the receiver is used. We study the problem of minimizing the AoII subject to a power constraint. By casting the problem into a \textit{CMDP}, we can prove that the optimal policy is a mixture of two deterministic threshold policies. Then, an efficient algorithm is proposed to find such policies and the mixing coefficient. Lastly, numerical results are provided to illustrate the performance of the AoII-optimal policy and compare it with the AoI-optimal policy.

\bibliographystyle{IEEEtran}
\bibliography{mybib}

\iffalse

\fi

\appendices
\section{Proof of Lemma \ref{le-IncreasingV}}\label{proof-IncreasingV}
To better distinguish between different states of the system, we denote by $V_\nu(x_d,x_{\Delta})$ the estimated value function of state $x$ at iteration $\nu$. To show the desired results, it is sufficient to prove that, at any iteration $\nu>0$, the following holds
\begin{subequations}\label{eq-IterativeRelop}
    \begin{empheq}[left=\empheqlbrace\,]{align}
      & V_\nu(d,\Delta_1)> V_\nu(d,\Delta_2) \ \ \ \forall\ \Delta_1>\Delta_2\geq0,\label{eq-InterativeDelta}\\
      & V_\nu(d_1,\Delta)> V_\nu(d_2,\Delta) \ \ \ \forall\ N-1\geq d_1>d_2\geq0.\label{eq-InterativeD}
    \end{empheq}
\end{subequations}
Leveraging the iterative nature of \textit{RVI}, we use induction to prove the desired results. Without loss of generality, we choose $x=(0,0)$ as the reference state. Since we initialize $V_0(d,\Delta)=\Delta$, \eqref{eq-IterativeRelop} holds when $\nu=0$. We suppose it holds up till iteration $\nu=t$ and examine whether it still holds at iteration $\nu=t+1$.

We first notice that the transition probabilities which dictate the structure of Bellman update depend only on $d$. Combining with the monotonic property of $V_{\nu}(\cdot)$, we conclude that \eqref{eq-InterativeDelta} holds at iteration $\nu+1$.

We next show the relationship between $V_{\nu+1}(d_1,\Delta)$ and $V_{\nu+1}(d_2,\Delta)$. To this end, we define $V_{\nu+1}^0(\cdot)$ and $V_{\nu+1}^1(\cdot)$ as the estimated value function if action $a=0$ and $a=1$ is chosen, respectively. Hence, we can combine and rewrite the Bellman update reported in \eqref{eq-BellmanRelative} and \eqref{eq-BellmanUpdate} as follows.
\begin{equation}\label{eq-BellmanAlternative}
V_{\nu+1}(d,\Delta) = \min\left\{V_{\nu+1}^0(d,\Delta),V_{\nu+1}^1(d,\Delta)\right\},
\end{equation}
where $V_{\nu+1}^a(d,\Delta)$ is calculated by
\begin{equation}\label{eq-UpdateValueAction}
\begin{split}
V_{\nu+1}^a(d,\Delta) = & \Delta + \lambda a + (1-p_sa)\sum_{d'=0}^{N-1}P_{d,d'}V_{\nu}(d',\Delta') +\\
& ap_s\big((1-2p)V_{\nu}(0,0) + 2pV_{\nu}(1,1)\big)-Q_{\nu+1}(x^{ref}),
\end{split}
\end{equation}
where $\Delta'=\mathbbm{1}_{\{d'\neq0\}}\times (\Delta+d')$ and $P_{d,d'}$ is specified in \eqref{eq-TransitionPro}. With this in mind, we divide our discussion into the following cases.
\begin{itemize}
	\item \textit{$d_1 = 1$ and $d_2=0$}: According to \eqref{eq-AoII}, $\Delta=0$ if and only if $d=0$. Then, we only need to compare $V_{\nu+1}(1,\Delta)$ with $V_{\nu+1}(0,0)$. Applying \eqref{eq-TransitionPro} to \eqref{eq-UpdateValueAction}, we arrive at the following results.
\begin{equation*}
V_{\nu+1}^0(1,\Delta)-V_{\nu+1}^0(0,0) =\Delta + \kappa_1,
\end{equation*}
\begin{equation*}
V_{\nu+1}^1(1,\Delta)-V_{\nu+1}^1(0,0) = \Delta+p_f\kappa_1,
\end{equation*}
where $\Delta>0$ and
\begin{equation*}
\begin{split}
\kappa_1 = & (1-3p)[V_{\nu}(1,\Delta+1)-V_{\nu}(0,0)]+p[V_{\nu}(2,\Delta+2)-V_{\nu}(1,1)]+\\
& p[V_{\nu}(1,\Delta+1) - V_{\nu}(1,1)].
\end{split}
\end{equation*}
	\item \textit{$d_1=2$ and $d_2=1$}: We need to compare $V_{\nu+1}(2,\Delta)$ with $V_{\nu+1}(1,\Delta)$. Following the same trajectory, we have the following.
	\begin{equation*}
	V_{\nu+1}^0(2,\Delta)-V_{\nu+1}^0(1,\Delta) = \kappa_2,
	\end{equation*}
	\begin{equation*}
	V_{\nu+1}^1(2,\Delta)-V_{\nu+1}^1(1,\Delta) = p_f\kappa_2,
	\end{equation*}
	where
	\begin{equation*}
	\begin{split}
	\kappa_2 = & p[V_{\nu}(1,\Delta+1)-V_{\nu}(0,0)]+p[V_{\nu}(3,\Delta+3)-V_{\nu}(2,\Delta+2)] +\\
	& (1-2p)[V_{\nu}(2,\Delta+2)-V_{\nu}(1,\Delta+1)].
	\end{split}
	\end{equation*}
	\item \textit{$2\leq d_2<d_1\leq N-2$}: We need to compare $V_{\nu+1}(d_1,\Delta)$ with $V_{\nu+1}(d_2,\Delta)$. Following again the same trajectory, we have the following.
	\begin{equation*}
	V_{\nu+1}^0(d_1,\Delta) - V_{\nu+1}^0(d_2,\Delta) = \kappa_3,
	\end{equation*}
	\begin{equation*}
	V_{\nu+1}^1(d_1,\Delta) - V_{\nu+1}^1(d_2,\Delta) = p_f\kappa_3,
	\end{equation*}
	where
	\begin{equation*}
	\begin{split}
	\kappa_3 = & (1-2p)[V_{\nu}(d_1,\Delta+d_1)-V_{\nu}(d_2,\Delta+d_2)]+\\
	& p[V_{\nu}(d_1-1,\Delta+d_1-1)-V_{\nu}(d_2-1,\Delta+d_2-1)]+\\
	& p[V_{\nu}(d_1+1,\Delta+d_1+1)-V_{\nu}(d_2+1,\Delta+d_2+1)].
	\end{split}
	\end{equation*}
	\item \textit{$d_1=N-1$ and $d_2=N-2$}: We need to compare $V_{\nu+1}(N-1,\Delta)$ with $V_{\nu+1}(N-2,\Delta)$. Following again the same trajectory, we have the following.
	\begin{equation*}
	V_{\nu+1}^0(N-1,\Delta) - V_{\nu+1}^0(N-2,\Delta) = \kappa_4,
	\end{equation*}
	\begin{equation*}
	V_{\nu+1}^1(N-1,\Delta) - V_{\nu+1}^1(N-2,\Delta) = p_f\kappa_4,
	\end{equation*}
	where
	\begin{equation*}
	\begin{split}
	\kappa_4 = & p[V_{\nu}(N-2,\Delta+N-2)-V_{\nu}(N-3,\Delta+N-3)] + \\
	& (1-3p)[V_{\nu}(N-1,\Delta+N-1)-V_{\nu}(N-2,\Delta+N-2)].
	\end{split}
	\end{equation*}
\end{itemize}
Baring in mind the monotonicity of $V_{\nu}(d,\Delta)$ and $p\in[0,\frac{1}{3}]$, we can easily see that $\kappa_1$, $\kappa_2$, $\kappa_3$, and $\kappa_4$ are all positive. Since the estimated value function is updated following \eqref{eq-BellmanAlternative}, we can easily verify that \eqref{eq-InterativeD} holds at iteration $t+1$ which concludes our proof.

\section{Proof of Proposition \ref{prop-ThresholdPolicy}}\label{proof-OptimalPolicy}
We continue with the same notations as in the proof of Lemma \ref{le-IncreasingV}. We recall that \textit{RVI} is an iterative algorithm and the estimated value function will converge to the value function. Hence, it is sufficient to show that the properties hold for the optimal policy at any iteration of \textit{RVI}.

We define $\delta V_{\nu}(d,\Delta) = V_{\nu}^1(d,\Delta) - V_{\nu}^0(d,\Delta)$.  Without loss of generality, we assume $t>0$. Then, the optimal action at iteration $\nu$ is captured by the sign of $\delta V_{\nu}(d,\Delta)$. More precisely, the optimal action $a_t^*=1$ if $\delta V_{\nu}(d,\Delta)\leq0$ and $a_t^*=0$ otherwise. Then, we can prove the following lemma.
\begin{lemma}\label{le-DeltaV1}
$\delta V_{\nu}(d,\Delta)$ is decreasing in $\Delta$ when $d\neq0$ and $t>0$.
\end{lemma}
\begin{proof}
We distinguish between following cases.
\begin{itemize}
	\item When $d=1$, applying \eqref{eq-TransitionPro} to \eqref{eq-UpdateValueAction} yields
	\begin{equation}\label{eq-Delta1}
	\begin{split}
	\delta V_{\nu}(1,\Delta) = & \lambda + p_s\{p[V_{\nu-1}(1,1)-V_{\nu-1}(1,\Delta+1)]+\\
	& (1-3p)[V_{\nu-1}(0,0)-V_{\nu-1}(1,\Delta+1)]+\\
	& p[V_{\nu-1}(1,1)-V_{\nu-1}(2,\Delta+2)]\}.
	\end{split}
	\end{equation}
	We notice that $(1-3p)$ is non-negative as $p\in[0,\frac{1}{3}]$.
	\item When $2\leq d\leq N-2$, following the same trajectory, we have
	\begin{equation}\label{eq-DeltaGeneral}
	\begin{split}
	\delta V_{\nu}(d,\Delta) = & \lambda+ p_s\{(1-2p)[V_{\nu-1}(0,0)-V_{\nu-1}(d,\Delta+d)]+\\
	& p[V_{\nu-1}(1,1)-V_{\nu-1}(d-1,\Delta+d-1)] +\\
	& p[V_{\nu-1}(1,1)-V_{\nu-1}(d+1,\Delta+d+1)]\}.
	\end{split}
	\end{equation}
	\item When $d=N-1$, following again the same trajectory, we have
	\begin{equation}\label{eq-DeltaLast}
	\begin{split}
	\delta V_{\nu}(N-1,\Delta) = & \lambda + p_s\{(1-2p)[V_{\nu-1}(0,0)-V_{\nu-1}(N-1,\Delta+N-1)] + \\
	& 2p[V_{\nu-1}(1,1)-V_{\nu-1}(N-2,\Delta+N-2)]\}.
	\end{split}
	\end{equation}
\end{itemize}
We recall that $\lambda$ is a non-negative constant and $V_{\nu-1}(d,\Delta)$ is increasing in both $d$ and $\Delta$ by Lemma \ref{le-IncreasingV}. Then, we can see that \eqref{eq-Delta1}, \eqref{eq-DeltaGeneral}, and \eqref{eq-DeltaLast} are nothing but the sum of a constant and a negative term that is decreasing in $\Delta$. Combing together, we can conclude our proof.
\end{proof}
With the lemma given, we can see that, for fixed $d\neq0$, $\delta V_{\nu}(d,\Delta)$ will decrease as $\Delta$ increases and, at some point, it will become negative. Therefore, for the states with fixed $d\neq0$, the optimal action $a_t^*$ will switch from $a_t^*=0$ to $a_t^*=1$ as $\Delta$ increases.\footnote{It is worth noting that $\delta V_{\nu}(d,\Delta)$ can always be negative which means that the optimal action $a_t^*$ will always be $a_t^*=1$. } We define the switching point for each $d\neq0$ as the first $\Delta$ such that $\delta V_{\nu}(d,\Delta)$ is non-positive. Since the instant cost is unbounded, the value function must also be unbounded. Therefore, the switching points always exist. We notice that the expressions of $\delta V_{\nu}(d,\Delta)$ differ for different $d$. Consequently, the corresponding switching points will also be different. To investigate the relationships between the switching points, we provide the following lemma.
\begin{lemma}\label{le-DeltaV2}
$\delta V_{\nu}(d,\Delta)$ is decreasing in $d$ when $d\neq0$ and $\nu>0$.
\end{lemma}
\begin{proof}
It is equivalent to show that $\forall\ \Delta>0$, $\delta V_{\nu}(d_1,\Delta)>\delta V_{\nu}(d_2,\Delta)$ if $1\leq d_1<d_2\leq N-1$. To this end, we distinguish between the following cases.
\begin{itemize}
\item When $d_1=1$ and $d_2=2$, leveraging \eqref{eq-Delta1} and \eqref{eq-DeltaGeneral}, we have
\begin{equation}\label{eq-DiffofDeltaV1}
\begin{split}
\delta V_{\nu}(1,\Delta)-\delta V_{\nu}(2,\Delta) = & p_s\{(1-2p)[V_{\nu-1}(2,\Delta+2)-V_{\nu-1}(1,\Delta+1)]+\\
& p[V_{\nu-1}(3,\Delta+3)-V_{\nu-1}(2,\Delta+2)]+\\
& p[V_{\nu-1}(1,\Delta+1)-V_{\nu-1}(0,0)]\}.
\end{split}
\end{equation}
\item When  $2\leq d_1<d_2\leq N-2$, leveraging \eqref{eq-DeltaGeneral}, we have
\begin{equation}\label{eq-DiffofDeltaV2}
\begin{split}
\delta V_{\nu}(d_1,\Delta) - \delta V_{\nu}(d_2,\Delta) & = p_s\{(1-2p)[V_{\nu-1}(d_2,\Delta+d_2)-V_{\nu-1}(d_1,\Delta+d_1)]+\\
& p[V_{\nu-1}(d_2-1,\Delta+d_2-1)-V_{\nu-1}(d_1-1,\Delta+d_1-1)]+\\
& p[V_{\nu-1}(d_2+1,\Delta+d_2+1)-V_{\nu-1}(d_1+1,\Delta+d_1+1)]\}.
\end{split}
\end{equation}
\item Similarly, when $d_1=N-2$ and $d_2=N-1$, we have
\begin{equation}\label{eq-DiffofDeltaV3}
\begin{split}
\delta V_{\nu}(N-2,\Delta)& - \delta V_{\nu}(N-1,\Delta) = \\
& p_s\{(1-3p)[V_{\nu-1}(N-1,\Delta+N-1)-V_{\nu-1}(N-2,\Delta+N-2)] + \\
& p[V_{\nu-1}(N-2,\Delta+N-2)-V_{\nu-1}(N-3,\Delta+N-3)]\}.
\end{split}
\end{equation}
\end{itemize}
According to Lemma \ref{le-IncreasingV}, $V_{\nu-1}(d,\Delta)$ is increasing in both $d$ and $\Delta$. Combining with the fact that $p\in[0,\frac{1}{3}]$, we can easily verify that \eqref{eq-DiffofDeltaV1}, \eqref{eq-DiffofDeltaV2}, and \eqref{eq-DiffofDeltaV3} are all positive. Consequently, $\delta V_{\nu}(d_1,\Delta)>\delta V_{\nu}(d_2,\Delta)$ holds $\forall\ 1\leq d_1<d_2\leq N-1$ which concludes our proof.
\end{proof}
Let $n^t_d$ denotes the switching point for the states with $d\neq0$ at iteration $t>0$. Then, $\delta V_{\nu}(d,n^t_d-1)>0$ and $\delta V_{\nu}(d,n^t_d)\leq0$. Since $\delta V_{\nu}(d,\Delta)$ is decreasing in $d$, $\delta V_{\nu}(d',n^t_d)<\delta V_{\nu}(d,n^t_d)\leq0$ if $d'>d$. This indicates that the ordering $n^t_{d'}\leq n^t_d$ must hold. Thus, we can conclude that the switching points $n^t_d$ when $d\neq0$ are non-increasing in $d$.

Finally, we discuss the only missing case: $d=0$. According to \eqref{eq-AoII}, $\Delta=0$ if and only if $d=0$. Thus, we only need to consider the state $(0,0)$. Then, we apply \eqref{eq-TransitionPro} to \eqref{eq-UpdateValueAction} which yields $\delta V_{\nu}(0,0)=\lambda$ for any $t>0$. As $\lambda$ is a non-negative constant, we can conclude that the optimal action at state $(0,0)$ is $a_t^*=0$.

As the above results are valid for any $\nu>0$ and \textit{RVI} converges to the value function as $\nu\rightarrow +\infty$ (i.e. $\lim_{\nu\rightarrow +\infty}V_{\nu}(\cdot)=V(\cdot)$), we can conclude that the above results are valid for the optimal policy for \eqref{eq-Minimization}.

\section{Proof of Theorem \ref{theo-ASMConvergence}}\label{proof-ASMConvergence}
We first introduce the infinite horizon $\gamma$-discounted cost of $\mathcal{M}$ where $0 < \gamma < 1$ is a discount factor. The expected $\gamma$-discounted cost under policy $\pi$ starting from state $x$ can be calculated as
\[
V_{\pi,\gamma}(x) = \mathbb{E}_{\pi}\left[\sum_{t=0}^{\infty}\gamma^tC(x_t,a_t)\ |\ x\right].
\]
The quantity $V_{\gamma}(\cdot) \triangleq inf_{\pi} V_{\pi,\gamma}(\cdot)$ is the best that can be achieved. Equivalently, $V_{\gamma}(\cdot)$ is  the value function associated with the infinite horizon $\gamma$-discounted \textit{MDP}. Then $V_{\gamma}(\cdot)$ satisfies the following Bellman equation.
\[
V_{\gamma}(x) = \min_{a}\left\lbrace C(x,a)+\gamma\sum_{x'\in\mathcal{X}}P_{xx'}(a)V_{\gamma}(x')\right\rbrace.
\]
We further define the quantity $v_{\gamma,n}(\cdot)$ as the minimum expected discounted cost for operating the system from time $t=0$ to $t = n-1$. It is known that $lim_{n\rightarrow\infty}v_{\gamma,n}(x) = V_{\gamma}(x)$, for all $x\in\mathcal{X}$. We also define the expected cost under policy $\pi$ starting from state $x$ as
\[
J_{\pi}(x) = \limsup_{n\rightarrow\infty}\frac{1}{n}\mathbb{E}_{\pi}\left[\sum_{t=0}^{n-1}C(x_t,a_t)\ |\ x\right],
\]
and $J(\cdot) \triangleq \inf_{\pi} J_{\pi}(\cdot)$ is the best that can be achieved. $V_{\pi,\gamma}^{(m)}(\cdot)$, $V_{\gamma}^{(m)}(\cdot)$, $v_{\gamma,n}^{(m)}(\cdot)$, $J^{(m)}_{\pi}(\cdot)$ and $J^{(m)}(\cdot)$ are defined analogously for the truncated \textit{MDP} $\mathcal{X}^{(m)}$. We define $h_{\gamma}^{(m)}(x)\triangleq V_{\gamma}^{(m)}(x)- V_{\gamma}^{(m)}(0)$ as the relative value function and chose the reference state $0=(0,0)$. For the simplicity of notation, for any two state $x,y\in\mathcal{X}$, we say $x\leq y$ if and only if $x_d\leq y_d$ and $x_\Delta\leq y_\Delta$.

With the above definitions in mind, we claim that our system verifies the two assumptions given in \cite{b17}. That is
\begin{itemize}
\item \textit{Assumption 1}: There exists a non-negative (finite) constant $L$, a non-negative (finite) function $M(\cdot)$ on $\mathcal{X}$, and constants $m_0$ and $\gamma_0\in [0, 1)$, such that $-L \leq h_{\gamma}^{(m)}(x) \leq M(x)$, for $x\in \mathcal{X}^{(m)}$, $m \geq m_0$, and $\gamma\in (\gamma_0, 1)$: We recall that $V_{\gamma}^{(m)}(x)$ is the value function and satisfies the Bellman equation. Thus, we can show that $V_{\gamma}^{(m)}(x)$ is increasing in $x$ in a similar way as did in Lemma \ref{le-IncreasingV}. The proof is omitted for the sake of space. Then, $h_{\gamma}^{(m)}(x)= V_{\gamma}^{(m)}(x)- V_{\gamma}^{(m)}(0)\geq0$. Consequently, we can choose $L=0$.

Let $c_{x,0}(\psi)$ be the expected cost of a first passage from $x\in\mathcal{X}$ to the reference state 0 when policy $\psi$ is adopted and $c_{x,0}^{(m)}(\psi)$ is defined analogously for the truncated \textit{MDP} $\mathcal{X}^{(m)}$. In the following, we consider the policy $\psi$ being always update policy where the transmitter makes transmission attempt at each time slot. Since the policy $\psi$ induces an irreducible ergodic Markov chain and the expected cost is finite, $h_{\gamma}^{(m)}(x)\leq c_{x,0}^{(m)}(\psi)$ from Proposition 5 of \cite{b19} and $c_{x,0}(\psi)$ is finite from Proposition 4 of \cite{b19}. We also know that $c_{x,0}(\psi)$ satisfies the following equation \cite{b17}.
\begin{equation}\label{eq-c}
c_{x,0}(\psi) = C(x,a^{\psi}) + \sum_{x'\in\mathcal{X}-\{0\}}P_{xx'}(a^{\psi})c_{x',0}(\psi).
\end{equation}
We can verify in a similar way to the proof of Lemma \ref{le-IncreasingV} that $c_{x,0}(\psi)$ is increasing in $x$. The proof is omitted here for the sake of space. Then, we obtain
\begin{equation}\label{eq-ASMInequality}
\begin{split}
\sum_{y\in\mathcal{X}^{(m)}_{-1}}P_{xy}^{(m)}(a^{\psi})c_{y,0}(\psi) & = \sum_{y\in\mathcal{X}^{(m)}_{-1}}P_{xy}(a^{\psi})c_{y,0}(\psi)+\sum_{y\in\mathcal{X}^{(m)}_{-1}}\left(\sum_{z\notin\mathcal{X}^{(m)}}P_{xz}(a^{\psi})q_z(y)\right)c_{y,0}(\psi)\\
& = \sum_{y\in\mathcal{X}^{(m)}_{-1}}P_{xy}(a^{\psi})c_{y,0}(\psi)+\sum_{z\notin\mathcal{X}^{(m)}}P_{xz}(a^{\psi})\left(\sum_{y\in\mathcal{X}^{(m)}_{-1}}q_z(y)c_{y,0}(\psi)\right)\\
& \leq \sum_{y\in\mathcal{X}^{(m)}_{-1}}P_{xy}(a^{\psi})c_{y,0}(\psi)+\sum_{z\notin\mathcal{X}^{(m)}}P_{xz}(a^{\psi})c_{z,0}(\psi)\\
& = \sum_{y\in\mathcal{X}-{\{0\}}}P_{xy}(a^{\psi})c_{y,0}(\psi),
\end{split}
\end{equation}
where $\mathcal{X}^{(m)}_{-1} = \mathcal{X}^{(m)}-\{0\}$. Applying \eqref{eq-ASMInequality} to \eqref{eq-c} yields
\begin{equation*}
c_{x,0}(\psi) \geq C(x,a^{\psi}) + \sum_{y\in\mathcal{X}^{(m)}-\{0\}}P_{xy}^{(m)}(a^{\psi})c_{y,0}(\psi).
\end{equation*}
Bearing in mind that $c_{x,0}^{(m)}(\psi)$ satisfies the following.
\begin{equation*}
c_{x,0}^{(m)}(\psi) = C(x,a^{\psi}) + \sum_{y\in\mathcal{X}^{(m)}-\{0\}}P_{xy}^{(m)}(a^{\psi})c_{y,0}^{(m)}(\psi),
\end{equation*}
we can conclude that $c_{x,0}^{(m)}(\psi)\leq c_{x,0}(\psi)$. Thus, we can choose $M(x)=c_{x,0}(\psi)<\infty$.

\item \textit{Assumption 2}: $\limsup_{m\rightarrow\infty}J^{(m)} \triangleq J^*<\infty$ and $J^*\leq J(x)$, for all $x\in\mathcal{X}$: We first show the hypothesis in Proposition 5.1 of \cite{b17} is true. Since we redistribute the excess probabilities in a way such that, for all $z\in\mathcal{X}-\mathcal{X}^{(m)}$,
\begin{equation*}
\sum_{y\in\mathcal{X}^{(m)}}q_z(y)v_{\gamma,n}(y) = v_{\gamma,n}(x),
\end{equation*}
where $x_d=z_d$ and $x_\Delta=m$, we only need to verify that, for all $z\in\mathcal{X} - \mathcal{X}^{(m)}$
\begin{equation}\label{eq-inequalityv}
v_{\gamma,n}(x) \leq v_{\gamma,n}(z).
\end{equation}
As $v_{\gamma,n}(x)$ adopts the following inductive form \cite{b17}.
\begin{equation*}
v_{\gamma,n+1}(x) = \min_{a}\left\lbrace C(x,a)+\gamma\sum_{x'\in\mathcal{X}}P_{xx'}(a)v_{\gamma,n}(x')\right\rbrace,
\end{equation*}
we can prove \eqref{eq-inequalityv} is true in a similar way to Lemma \ref{le-IncreasingV} and the proof is omitted for the sake of space. $J(x)$ is trivially finite for $x\in\mathcal{X}$. Then, according to Corollary 5.2 of \cite{b17}, assumption 2 is valid.
\end{itemize}
Consequently, the following results are true.
\begin{itemize}
\item There exists an average cost optimal stationary policy in $\mathcal{M}^{(m)}$.
\item Any limit point of the sequence of optimal policies in $\mathcal{M}^{(m)}$ is optimal in $\mathcal{M}$.
\end{itemize}

\section{Proof of Proposition \ref{prop-TransmissionRate}}\label{proof-TransmissionRate}
We first delve into the state space $\mathcal{X}$ of the \textit{MDP} $\mathcal{M}$ and provide the condition it must satisfy. Without loss of generality, we suppose the system always starts from state $(0,0)$. We claim that state $(d,\Delta)$ with $d\neq0$ must satisfy the following condition.
\begin{equation}\label{eq-BoundDelta}
\Delta \geq l_d = \sum_{i=1}^{d}i = \frac{d^2+d}{2}.
\end{equation}
To see the condition, we notice that the transition to state $(0,0)$ is equivalent to restarting the system. Thus, it is sufficient to consider the sequence of transitions starting from the last time the system is at state $(0,0)$. Therefore, the age $\Delta$ will always increase. We recall that the maximum jump of $d$ is 1 as specified in \eqref{eq-TransitionPro}. Thus, we can conclude that there always exists a lower bound $l_{d}$ on the age $\Delta$ for any given $d\neq0$. Combing with the system dynamic discussed in Section \ref{sec-SystemDynamic}, the lower bound in \eqref{eq-BoundDelta} is easy to obtain. To make the structure of equations consistent, we define the states that violate the condition \eqref{eq-BoundDelta} as \textit{virtual} states since the system will never reach these states.

As the state space is clarified, we proceed with deriving the main results. We first recall that the threshold policy $\bm{n}$ possesses the properties detailed in Proposition \ref{prop-ThresholdPolicy}. More precisely, for the state with given $d\neq0$, the action suggested by $\bm{n}$ is $a^*=1$ if the age $\Delta$ is larger than or equal to the corresponding threshold $n_d$. Hence, we define $\tau=max\{\bm{n}\}$. An important property of $\tau$ is that, for the states with $\Delta\geq\tau$, the actions suggested by $\bm{n}$ are the same. We define $a_{d,\Delta}$ as the action suggested by $\bm{n}$ at state $(d,\Delta)$. For each $1\leq d\leq N-1$, we define 
\[
\Pi_d(\tau) = \sum_{\Delta=\tau}^{+\infty}a_{d,\Delta}\pi_d(\Delta) = \sum_{\Delta=\tau}^{+\infty}\pi_d(\Delta).
\]
The last equality holds since $a_{d,\Delta}=1$ for all the states with $\Delta\geq\tau$. Then, the expected transmission rate can be calculated as
\begin{equation*}
\bar{R}_{\bm{n}} = \sum_{d=1}^{N-1}\sum_{\Delta=n_d}^{+\infty}\pi_d(\Delta) = \sum_{d=1}^{N-1}\left(\sum_{\Delta=n_d}^{\tau-1}\pi_d(\Delta) + \Pi_d(\tau)\right).
\end{equation*}
We claim that $\Pi_d(\tau)$'s, along with the stationary distribution $\pi_d(\Delta)$'s, can be obtained by solving a finite system of linear equations induced from the balance equation \eqref{eq-BalanceEquation}. Leveraging the results in Section \ref{sec-SystemDynamic}, we distinguish between the following cases.
\begin{itemize}
\item For state $(0,0)$, \eqref{eq-BalanceEquation} can be written as
\begin{equation*}
\begin{split}
\pi_0(0) & = (1-2p)\pi_0(0) + p\sum_{\Delta=1}^{+\infty}(1-p_sa_{1,\Delta})\pi_1(\Delta)+p_s(1-2p)\sum_{d=1}^{N-1}\sum_{\Delta=l_d}^{+\infty}a_{d,\Delta}\pi_d(\Delta)\\
& = (1-2p)\pi_0(0) + p\sum_{\Delta=1}^{\tau-1}(1-p_sa_{1,\Delta})\pi_1(\Delta)+p_fp\Pi_1(\tau) + \\
&\ \ \ \ p_s(1-2p)\sum_{d=1}^{N-1}\Bigg(\sum_{\Delta=n_d}^{\tau-1}\pi_d(\Delta)+\Pi_d(\tau)\Bigg).
\end{split}
\end{equation*}
This recovers \eqref{eq-Balance0}.
\item For state $(1,1)$, \eqref{eq-BalanceEquation} can be written as
\begin{equation*}
\begin{split}
\pi_1(1)  &= 2p\pi_0(0) + 2p_sp\sum_{d=1}^{N-1}\sum_{\Delta=l_d}^{+\infty}a_{d,\Delta}\pi_d(\Delta)\\
&= 2p\pi_0(0) + 2p_sp\sum_{d=1}^{N-1}\left(\sum_{\Delta=n_d}^{\tau-1}\pi_d(\Delta)+\Pi_d(\tau)\right).\\
\end{split}
\end{equation*}
This recovers \eqref{eq-Balance1}. 
\item For the virtual states, we define the steady-state probabilities as zero since the system will never reach these states. This recovers the first equation of \eqref{eq-BalanceGeneral}.
\item For other states, leveraging the definition of virtual states, we obtain an alternative form of \eqref{eq-BalanceEquation} which is
\begin{equation}\label{eq-GeneralPiCalculate}
\pi_d(\Delta) = \sum_{d'=1}^{N-1}P_{d',d}(1-p_sa_{d',\Delta-d})\pi_{d'}(\Delta-d).
\end{equation}
As $\Delta\in\mathbbm{N}^*$, there are infinitely many equations to solve. Inspired by the definition of $\Pi_d(\tau)$, we can combine the states with $\Delta\geq\tau$ and eliminate the infinity. More precisely, for each $1\leq d\leq N-1$, we do the following.
\begin{equation}\label{eq-PtauCalculate}
\begin{split}
\sum_{\Delta=\tau}^{+\infty}\pi_d(\Delta) & = \Pi_d(\tau) = \sum_{d'=1}^{N-1}P_{d',d}\left(\sum_{\Delta=\tau}^{+\infty}(1-p_sa_{d',\Delta-d})\pi_{d'}(\Delta-d)\right)\\
& = \sum_{d'=1}^{N-1}P_{d',d}\Bigg(\sum_{\Delta=\tau-d}^{\tau-1}(1-p_sa_{d',\Delta})\pi_{d'}(\Delta)+ p_f\sum_{\Delta=\tau}^{+\infty}\pi_{d'}(\Delta)\Bigg)\\
& = \sum_{d'=1}^{N-1}P_{d',d}\Bigg(\sum_{\Delta=\tau-d}^{\tau-1}(1-p_sa_{d',\Delta})\pi_{d'}(\Delta)+ p_f\Pi_{d'}(\tau)\Bigg).
\end{split}
\end{equation}
Combining \eqref{eq-GeneralPiCalculate} and \eqref{eq-PtauCalculate}, we recover the second and third equation of \eqref{eq-BalanceGeneral}.
\end{itemize}
Equation \eqref{eq-BalanceSum} is obtained from the fact that the steady-state probabilities must add up to one.

By combining the states with $\Delta\geq\tau$, we actually cast the induced infinite-state Markov chain to a finite-state Markov chain. Therefore, the expected transmission rate can be calculated theorecially without any approximation.

\section{Proof of Corollary \ref{prop-RateApproximation}}\label{proof-RateApproximation}
We inherit the notations and definitions from the proof of Proposition \ref{prop-TransmissionRate}. Before presenting the main results, we first introduce the key approximation used in the derivation of the main results. We note that when the thresholds in $\bm{n}$ are huge, the expected transmission rate will be insignificant. More precisely, when the thresholds $n_d$'s are huge,
\begin{equation*}
\bar{R}_{\bm{n}}=\sum_{d=1}^{N-1}\left(\sum_{\Delta=n_d}^{\tau-1}\pi_d(\Delta)+\Pi_d(\tau)\right)\approx 0.
\end{equation*}
We apply the above equation to \eqref{eq-Balance1} and obtain the following key approximation.
\begin{equation}\label{eq-Approximation}
\pi_1(1) \approx 2p\pi_0(0).
\end{equation}
Then, we claim that, for any state $(d,\Delta)$, the steady-state probability $\pi_d(\Delta)$ can be approximated as
\[
\pi_d(\Delta) \approx c^{\bm{n}}_{d,\Delta}\pi_0(0),
\]
where $c^{\bm{n}}_{d,\Delta}$ is a scalar depends on the policy and the state. To prove this, we first recall that the transitions in the induced Markov chain always go along the increasing direction of $\Delta$ unless it goes back to state $(0,0)$ or $(1,1)$. Combining with the approximation made in \eqref{eq-Approximation}, we can see that any $\pi_d(\Delta)$ can be approximated as a multiple of $\pi_0(0)$. 

With this in mind, we notice that, for any two threshold policies $\bm{n}_1$ and $\bm{n}_2$, the suggested actions by the two polices at states with $\Delta<\min\{[\bm{n}_1,\bm{n}_2]\}$ are the same (i.e. $a_{d,\Delta}=0$). We denote by $G(\bm{n}_1,\bm{n}_2)$ the set of these states. Then, for state $(d,\Delta)\in G(\bm{n}_1,\bm{n}_2)$, regardless of whether $\bm{n}_1$ or $\bm{n}_2$ is adopted, the balance equation is the same. Consequently, the corresponding $c^{\bm{n}}_{d,\Delta}$ is independent of policy. Then, for $(d,\Delta)\in G(\bm{n}_1,\bm{n}_2)$,
\begin{equation}\label{eq-ApproximationUsed}
\frac{\pi^1_d(\Delta)}{\pi^2_d(\Delta)}\approx\frac{c_{d,\Delta}\pi^1_0(0)}{c_{d,\Delta}\pi^2_0(0)} = \frac{\pi^1_0(0)}{\pi^2_0(0)},
\end{equation}
where $\pi^1_d(\Delta)$'s and $\pi^2_d(\Delta)$'s are the stationary distribution when $\bm{n}_1$ and $\bm{n}_2$ is adopted, respectively. Leveraging \eqref{eq-ApproximationUsed}, we can obtain the main results in the corollary. 
\iffalse
To this end, we claim that the expected transmission rate can be approximated as
\begin{equation*}
\bar{R}_{\bm{n}} \approx \sum_{d=1}^{N-1}\left(\sum_{\Delta=n_d}^{\tau-1}\pi_d(\Delta) + \Pi_d(\tau)\right),
\end{equation*}
where $\Pi_d(\tau)$'s along with the stationary distribution $\pi_d(\Delta)$'s can be obtained by solving a finite system of linear equations induced from the original infinite system of linear equations \eqref{eq-BalanceEquation}. To see the approximation, 
\fi
We first define $\eta=min\{\bm{n}\}-1$ and $\Pi_d(\eta)=\sum_{\Delta=1}^{\eta}\pi_d(\Delta)$. Similar to what we did in the proof of Proposition \ref{prop-TransmissionRate}, we rewrite the balance equation \eqref{eq-BalanceEquation} as follows.
\begin{itemize}
\item For state $(0,0)$, \eqref{eq-BalanceEquation} can be written as
\begin{equation*}
\begin{split}
\pi_0(0)  &= (1-2p)\pi_0(0) + p\sum_{\Delta=1}^{+\infty}(1-p_sa_{1,\Delta})\pi_1(\Delta)+p_s(1-2p)\sum_{d=1}^{N-1}\sum_{\Delta=l_d}^{+\infty}a_{d,\Delta}\pi_d(\Delta)\\
&= (1-2p)\pi_0(0) + p\sum_{\Delta=\eta+1}^{\tau-1}(1-p_sa_{1,\Delta})\pi_1(\Delta)+p\Pi_1(\eta)+p_fp\Pi_1(\tau) + \\
& \ \ \ \ p_s(1-2p)\sum_{d=1}^{N-1}\Bigg(\sum_{\Delta=n_d}^{\tau-1}\pi_d(\Delta)+\Pi_d(\tau)\Bigg).
\end{split}
\end{equation*}
This recovers \eqref{eq-Approximation0}.
\item For state $(1,1)$, \eqref{eq-BalanceEquation} can be written as
\begin{equation*}
\begin{split}
\pi_1(1)  &= 2p\pi_0(0) + 2p_sp\sum_{d=1}^{N-1}\sum_{\Delta=l_d}^{+\infty}a_{d,\Delta}\pi_d(\Delta)\\
&= 2p\pi_0(0) + 2p_sp\sum_{d=1}^{N-1}\left(\sum_{\Delta=n_d}^{\tau-1}\pi_d(\Delta)+\Pi_d(\tau)\right).
\end{split}
\end{equation*}
This recovers \eqref{eq-Approximation1}.
\item For other states, leveraging the definition of virtual states, we obtain an alternative form of \eqref{eq-BalanceEquation} which is
\begin{equation}\label{eq-ApproximationGeneral}
\pi_d(\Delta) = \sum_{d'=1}^{N-1}P_{d',d}(1-p_sa_{d',\Delta-d})\pi_{d'}(\Delta-d).
\end{equation}
\end{itemize}

Instead of applying  \eqref{eq-ApproximationGeneral} directly, we combine the states with $\Delta\leq\eta$ to reduce the number of equations. More precisely, for each $2\leq d\leq N-1$, we have
\begin{equation}\label{eq-ApproximationHead1}
\begin{split}
\sum_{\Delta=1}^{\eta+d}\pi_d(\Delta) & = \Pi_d(\eta) + \sum_{\Delta=\eta+1}^{\eta+d}\pi_d(\Delta) \\
& =\sum_{\Delta=1}^{\eta+d}\Bigg(\sum_{d'=1}^{N-1}P_{d',d}(1-p_sa_{d',\Delta-d})\pi_{d'}(\Delta-d)\Bigg)\\
& = \sum_{d'=1}^{N-1}P_{d',d}\Pi_{d'}(\eta).
\end{split}
\end{equation}
When $d=1$, due to the particularity of state $(1,1)$, we have
\begin{equation}\label{eq-ApproximationHead2}
\begin{split}
& \sum_{\Delta=2}^{\eta+1}\pi_1(\Delta) = \Pi_1(\eta)-\pi_1(1)+\sum_{\Delta=\eta+1}^{\eta+1}\pi_1(\Delta) = \sum_{d'=1}^{N-1}P_{d',1}\Pi_{d'}(\eta).
\end{split}
\end{equation}
We notice that \eqref{eq-ApproximationHead1} or \eqref{eq-ApproximationHead2} involves $\pi_d(\Delta)$'s where $1\leq d\leq N-1$ and $\eta+1\leq\Delta\leq\eta+d$. Under usual circumstances, these steady-state probabilities can be calculated using \eqref{eq-ApproximationGeneral}. However, $\pi_d(\Delta)$'s where $\Delta\leq\eta$ are required when applying \eqref{eq-ApproximationGeneral} and we have no access to them as we combined them together as $\Pi_d(\eta)$. To circumvent this, we use the approximation reported in \eqref{eq-ApproximationUsed}. More precisely, for the states with $1\leq d\leq N-1$ and $\eta+1\leq\Delta\leq\eta+d$, we have
\begin{equation}\label{eq-MainApproximation}
\begin{split}
\pi_d(\Delta) & =\sum_{d'=1}^{N-1}P_{d',d}(1-p_sa_{d',\Delta-d})\pi_{d'}(\Delta-d)\\
& \approx \rho\sum_{d'=1}^{N-1}P_{d',d}(1-p_sa_{d',\Delta-d})\sigma_{d'}(\Delta-d),
\end{split}
\end{equation}
where $\rho\triangleq\frac{\pi_0(0)}{\sigma_0(0)}$ and $\sigma_d(\Delta)$ is the stationary distribution of the Markov chain induced by another policy $\bm{n}'$ and can be calculated using Proposition \ref{prop-TransmissionRate}. In order to utilize the approximation reported in \eqref{eq-ApproximationUsed}, the policy $\bm{n}'$ must satisfy
\begin{equation*}
min\{[\bm{n},\bm{n}']\} > \eta.
\end{equation*}
We recall in Proposition \ref{prop-TransmissionRate}, the computational complexity of calculating the stationary distribution of a Markov chain induced by a threshold policy depends on the maximal threshold. To make the calculation of $\sigma_d(\Delta)$ as cheap as possible, we choose $\bm{n}'=[\eta',...,\eta']$ where $\eta'=\min\{\bm{n}\}$. Combining \eqref{eq-ApproximationHead1}, \eqref{eq-ApproximationHead2}, and \eqref{eq-MainApproximation}, we recover \eqref{eq-ApproximationHeads} and the first equation of \eqref{eq-ApproximationD}.

For state with $1\leq d\leq N-1$ and $\eta+d+1\leq \Delta\leq\tau-1$, leveraging the above approximation, we can calculate the steady-state probabilities using \eqref{eq-ApproximationGeneral} which recovers the second equation of \eqref{eq-ApproximationD}.

Finally, for state with $1\leq d\leq N-1$ and $\Delta\geq\tau$, we combine them as did in Proposition \ref{prop-TransmissionRate}. Then, we can recover the third equation of \eqref{eq-ApproximationD}.

Equation \eqref{eq-ApproximationSum} is obtained form the fact that the sum of all steady-state probabilities must be one.

By combining the states with $\Delta\leq\eta$, we reduce the size of the finite-state Markov chain cast to. Although approximation is used, as the thresholds increase, the approximation in \eqref{eq-Approximation} will become more and more accurate.

\section{Proof of Corollary \ref{prop-ExpectedAoII}}\label{proof-ExpectedAoII}
We still inherit the notations and definitions from the proof of Proposition \ref{prop-TransmissionRate}. We first recall that the AoII at state $(d,\Delta)$ is nothing but $\Delta$. Then, similar to what we did in the proof of Proposition \ref{prop-TransmissionRate}, the expected AoII under threshold policy $\bm{n}$ can be calculated as
\[
\bar{\Delta}_{\bm{n}} = \sum_{d=1}^{N-1}\left(\sum_{\Delta=l_d}^{\tau-1}\omega_d(\Delta) + \Omega_d(\tau)\right),
\]
where $\tau=max\{\bm{n}\}$, $l_d=\frac{d^2+d}{2}$, and
\[
\begin{split}
& \omega_d(\Delta) \triangleq \Delta\pi_d(\Delta),\\
& \Omega_d(\tau) \triangleq \sum_{\Delta=\tau}^{+\infty}\omega_d(\Delta).
\end{split}
\]
Note that $\pi_d(\Delta)$'s are the stationary distribution of the infinite-state Markov chain induced from the same threshold policy $\bm{n}$. We claim that $\Omega_d(\tau)$'s, along with $\omega_d(\Delta)$'s, can be obtained by solving a finite system of linear equations. To this end, we distinguish between the following cases.
\begin{itemize}
\item For the virtual states, we have $\omega_d(\Delta)=0$ because $\pi_d(\Delta)=0$ for these state by definition. Meanwhile, $\omega_0(0)=0$ because no cost is paid for being at state $(0,0)$. This recovers the first equation of \eqref{eq-ExpectedAoIIGeneral}.
\item For the states with $1\leq d\leq N-1$ and $l_d\leq\Delta\leq\tau-1$, we have
\begin{equation}\label{eq-ExpectedMiddle}
\omega_d(\Delta) = \Delta\pi_d(\Delta).
\end{equation}
This recovers the second equation of \eqref{eq-ExpectedAoIIGeneral}.
\item For the states with $1\leq d\leq N-1$ and $\Delta\geq\tau$, we can use \eqref{eq-ExpectedMiddle}. But in this case, we need to calculate an infinite number of values. To eliminate the infinity, we notice that $\omega_d(\Delta)$'s can also be calculated by multiplying both sides of \eqref{eq-GeneralPiCalculate} by $(\Delta-d)$. More precisely, we have
\begin{equation*}
\begin{split}
(\Delta-&d)\pi_d(\Delta) = \sum_{d'=1}^{N-1}P_{d',d}(1-p_sa_{d',\Delta-d})(\Delta-d)\pi_{d'}(\Delta-d).
\end{split}
\end{equation*}
Applying the definition of $\omega_d(\Delta)$, we have
\[
\begin{split}
\omega_d(\Delta) -& d\pi_d(\Delta) = \sum_{d'=1}^{N-1}P_{d',d}(1-p_sa_{d',\Delta-d})\omega_{d'}(\Delta-d).
\end{split}
\]
Like we did in the proof of Proposition \ref{prop-TransmissionRate}, we combine the states with $\Delta\geq\tau$ to eliminate the infinity. More precisely, for each $1\leq d\leq N-1$, we have
\[
\begin{split}
\sum_{\Delta=\tau}^{+\infty}\Big(\omega_d(\Delta) - d\pi_d(\Delta)\Big) & = \Omega_d(\tau) - d\Pi_d(\tau) = \sum_{d'=1}^{N-1}P_{d',d}\left(\sum_{\Delta=\tau}^{+\infty}(1-p_sa_{d',\Delta-d})\omega_{d'}(\Delta-d)\right)\\
& = \sum_{d'=1}^{N-1}P_{d',d}\Bigg(\sum_{\Delta=\tau-d}^{\tau-1}(1-p_sa_{d',\Delta})\omega_{d'}(\Delta)+ p_f\sum_{\Delta=\tau}^{+\infty}\omega_{d'}(\Delta)\Bigg)\\
& = \sum_{d'=1}^{N-1}P_{d',d}\Bigg(\sum_{\Delta=\tau-d}^{\tau-1}(1-p_sa_{d',\Delta})\omega_{d'}(\Delta)+ p_f\Omega_{d'}(\tau)\Bigg).
\end{split}
\]
This recovers the last equation of \eqref{eq-ExpectedAoIIGeneral}.
\end{itemize}

\section{Proof of Theorem \ref{theo-Construct}}\label{proof-Construct}
We first make the following definitions. When the \textit{MDP} $\mathcal{M}$ is at state $x$ and action $a$ is chosen, cost $C_1(x,a)=x_\Delta$ and $C_2(x,a)=\lambda a$ are incurred. We define the expected $C_1$-cost and the expected $C_2$-cost under policy $\pi$ as $\bar{C}_1(\pi)$ and $\bar{C_2}(\pi)$, respectively. Let $G$ be a nonempty set and $\mathcal{R}^*(i,G)$ be the class of policies $\pi$ such that 
\begin{itemize}
\item $P_{\pi}(x_n\in G\ for\ some\ n\geq1\ |\ x_0=i)=1$ where $x_n$ is the state of $\mathcal{M}$ at time $n$.
\item The expected time $m_{iG}(\pi)$ of a first passage from $i$ to $G$ under $\pi$ is finite.
\item The expected $C_1$-cost $\bar{C}_1^{i,G}(\pi)$ and the expected $C_2$-cost $\bar{C}_2^{i,G}(\pi)$ of a first passage form $i$ to $G$ under $\pi$ are finite. 
\end{itemize}
With the above definitions clarified, we proceed with presenting the assumptions given in \cite{b15} and verifying our system satisfies all the assumptions.
\begin{enumerate}
\item For all $w>0$, the set $G(w)=\{x\ | $ \textit{there exists an action a such that} $C_1(x,a)+C_2(x,a)\leq w\}$ is finite: For our system, we have $C_1(x,a)+C_2(x,a)=x_\Delta + \lambda a\geq x_\Delta$. Then, any state $x$ in $G(w)$ must satisfy $x_\Delta \leq w$. Bearing in mind that $x_\Delta\in\mathbb{N}_0$, we can conclude that, the set $G(w)$ is always finite.
\item There exists a stationary policy $e$ such that the induced Markov chain has the following properties: the state space $\mathcal{X}$ consists of a single (non-empty) positive recurrent class $R$ and a set $U$ of transient states such that $e\in\mathcal{R}^*(i,R)$, for $i\in U$. Moreover, both $\bar{C}_1(e)$ and $\bar{C}_2(e)$ on $R$ are finite: We consider the always update policy $\psi_{au}$ where the transmitter makes transmission attempt at every time slot. We take the set $R=\mathcal{X}$. Applying the system dynamic discussed in Section \ref{sec-SystemDynamic}, we can see that, under $\psi_{au}$, all the states in $R$ communicate with state $(0,0)$ and state $(0,0)$ is positive recurrent. Consequently, we can conclude that the set $R$ forms a positive recurrent class. The set $U$ can simply be empty set. Finally, we notice that $\bar{C}_2(\psi_{au})$ is nothing but the expected transmission rate which is finite and $\bar{C}_1(\psi_{au})$ is the expected AoII which is also finite.
\item Given any two state $x\neq y$, there exists a policy $\pi$ such that $\pi\in\mathcal{R}^*(x,y)$: We first notice that any state $x\in\mathcal{X}$ communicates with state $(0,0)$ with positive probability if the transmitter makes a transmission attempt at state $x$ and succeeds. We also notice that state $(0,0)$ can reach any state $x\in\mathcal{X}$ as the minimum increase in both $d$ and $\Delta$ is one. Consequently, we can always find a policy that induces a Markov chain such that there exists a path with a positive probability between any two different states $x$ and $y$. The corresponding $\bar{C}_1^{x,y}(\pi)$, $\bar{C}_2^{x,y}(\pi)$ and $m_{x,y}(\pi)$ are trivially finite.
\item If a stationary policy $\pi$ has at least one positive recurrent state, then it has a single positive recurrent class $R$. Moreover, if $x\notin R$, then $\pi\in\mathcal{R}^*(x,R)$ where $x=(0,0)$: We notice that, for any policy, the penalty can decrease only when the system reaches state $(0,0)$ or $(1,1)$. At the same time, $(0,0)$ and $(1,1)$ communicate with each other. Thus, any positive recurrent class must contain $(0,0)$ and $(1,1)$ which indicates that there can only be a single positive recurrent class.
\item There exists a policy $\pi$ such that $\bar{C}_1(\pi)<\infty$ and $\bar{C}_2(\pi)<\alpha$: We first note that $\bar{C}_2(\pi)$ is simply the expected transmission rate. Then, we can always find a policy with large enough thresholds such that $\bar{C}_2(\pi)$ is less than $\alpha$. We can easily verify that the corresponding $\bar{C}_1(\pi)$ is finite.
\end{enumerate}
Some other results in \cite{b15} will be useful when constructing the optimal policy, especially Proposition 3.2, Lemma 3.4, 3.7, 3.9 and 3.10. To this end, we define $\bar{R}_{\lambda}$ as the expected transmission rate associate with policy $\bm{n}_{\lambda}$ and $\lambda^*\triangleq\inf\{\lambda>0:\bar{R}_{\lambda}\leq\alpha\}$. We say a policy is $\lambda^*$-optimal if the policy is optimal for the \textit{MDP} $\mathcal{M}$ with $\lambda=\lambda^*$.

We know that there exists $\lambda^*_+\downarrow\lambda^*$ and $\lambda^*_-\uparrow\lambda^*$ such that they both converge to $\lambda^*$. At the same time, the corresponding optimal policies $\bm{n}_{\lambda^*_+}$ and $\bm{n}_{\lambda^*_-}$ will also converge and are both $\lambda^*$-optimal (Lemma 3.4 and 3.7 of \cite{b15}). Since the Markov chains induced by policies $\bm{n}_{\lambda^*_+}$ and $\bm{n}_{\lambda^*_-}$ are both irreducible and state $(0,0)$ is positive recurrent in both Markov chains, we can choose which policy to adopt every time the system reaches state $(0,0)$ independently without changing its optimality (Proposition 3.2 and Lemma 3.9 of \cite{b15}). Thus, we can mix the two policies in the following way: when the system reaches state $(0,0)$, the system will choose $\bm{n}_{\lambda^*_-}$ with probability $\mu$ and $\bm{n}_{\lambda^*_+}$ with probability $1-\mu$. Then the system will follow the chosen policy until the next choice. The probability $\mu$ is chosen such that the expected transmission rate of the mixed policy $\bm{n}_{\lambda^*}$ is equal to $\alpha$. More precisely,
\begin{equation*}
\mu = \frac{\alpha - \bar{R}_{\lambda^*_+}}{\bar{R}_{\lambda^*_-}-\bar{R}_{\lambda^*_+}}.
\end{equation*}
Then, we can conclude that the mixed policy $\bm{n}_{\lambda^*}$ is optimal for the constrained problem \eqref{eq-Constrained} (Lemma 3.10 of \cite{b15}).

\onecolumn  % switch to one column
\section{}\label{sec-Algorithm}
\begin{figure}[ht]
  \centering
  \begin{minipage}{0.7\columnwidth}
\begin{algorithm}[H]
    \begin{algorithmic}[1]
    \Require
    \Statex MDP $\mathcal{M} = (\mathcal{X},\mathcal{P},\mathcal{A},\mathcal{C})$
    \Statex Convergence Criteria $\epsilon$
    \Procedure{RelativeValueIteration}{$\mathcal{M}$}
        \State Initialize $V_0(x)=x_{\Delta}$; $\nu=0$
        \State Choose $x^{ref}\in\mathcal{X}$ arbitrarily
        \While{$V_{\nu}$ is not converged\footnotemark}
            \For{$x \in \mathcal{X}$}
                \If {$\exists$ active state $y$ s.t. $y_d\leq x_d$ and $y_\Delta\leq x_\Delta$}
                \State $a^*(x) = 1$
                \State $Q_{\nu+1}(x) = C(x,1) + \sum_{x'} P_{xx'}(1) \cdot V_{\nu}(x')$
                \Else
                \For{$a \in \mathcal{A}$}
                    \State $H_{x,a} = C(x, a) + \sum_{x'} P_{xx'}(a) \cdot V_{\nu}(x')$
                \EndFor
                \State $a^*(x) = \arg\min_a \{H_{x,a}\}$
                \State $Q_{\nu+1}(x) = H_{x,a^*}$
            \EndIf
            \EndFor
            \For{$x \in \mathcal{X}$}
                \State $V_{\nu+1}(x) = Q_{\nu+1}(x) - Q_{\nu+1}(x^{ref})$
            \EndFor
            \State $\nu = \nu + 1$
        \EndWhile
        \Return $\bm{n} \gets a^*(x)$
    \EndProcedure
    \end{algorithmic}
\caption{Improved Relative Value Iteration}
\label{alg-RVIA}
\end{algorithm}
\end{minipage}
\end{figure}
\footnotetext{\textit{RVI} converges when the maximum difference between the results of two consecutive iterations is less than $\epsilon$.}

\begin{figure}[ht]
  \centering
  \begin{minipage}{0.7\columnwidth}
\begin{algorithm}[H]
    \begin{algorithmic}[1]
    \Require
    \Statex Power Budget $\alpha$
    \Statex MDP $\mathcal{M}^{(m)}(\lambda) = (\mathcal{X}^{(m)},\mathcal{A},\mathcal{P}^{(m)},\mathcal{C}(\lambda))$
    \Statex Tolerance $\xi$
    \Procedure{BisectionSearch}{$\mathcal{M}^{(m)}(\lambda)$, $\alpha$}
        \State Initialize $\lambda_-=0$; $\lambda_+=1$
        \State $\bm{n}_{\lambda_+} = RVI(\mathcal{M}^{(m)}(\lambda_+),\epsilon)$ using Algorithm \ref{alg-RVIA}
        \State $\bar{R}_{\lambda_+} = \bar{R}(\bm{n}_{\lambda_+})$ using Proposition \ref{prop-TransmissionRate}
        \While{$\bar{R}_{\lambda_+}\geq\alpha$}
            \State $\lambda_-= \lambda_+$; $\lambda_+ = 2\lambda_+$
            \State $\bm{n}_{\lambda_+} = RVI(\mathcal{M}^{(m)}(\lambda_+),\epsilon)$ using Algorithm \ref{alg-RVIA}
        	\State $\bar{R}_{\lambda_+} = \bar{R}(\bm{n}_{\lambda_+})$ using Proposition \ref{prop-TransmissionRate}
        \EndWhile
        \While{$\lambda_+ - \lambda_- \geq \xi$}
            \State $\lambda = \frac{\lambda_+ + \lambda_-}{2}$
            \State $\bm{n}_{\lambda} = RVI(\mathcal{M}^{(m)}(\lambda),\epsilon)$ using Algorithm \ref{alg-RVIA}
        	\State $\bar{R}_{\lambda} = \bar{R}(\bm{n}_{\lambda})$ using Proposition \ref{prop-TransmissionRate}
            \If{$\bar{R}_{\lambda}\geq \alpha$}
                \State $\lambda_-=\lambda$
            \Else
                \State $\lambda_+=\lambda$
            \EndIf
        \EndWhile
        \Return $(\lambda_+^*,\lambda_-^*) \gets (\lambda_+,\lambda_-)$
    \EndProcedure
    \end{algorithmic}
\caption{Bisection Search}
\label{alg-BisectionSearch}
\end{algorithm}
\end{minipage}
\end{figure}

\end{document}